\documentclass[pdflatex,sn-mathphys-num]{sn-jnl}

\usepackage{graphicx}%
\usepackage{subcaption}
\usepackage{amsmath,amssymb,amsfonts}%
\usepackage{amsthm}%
\usepackage{mathrsfs}%
\usepackage{mathdots}
\usepackage[title]{appendix}%
\usepackage{longtable}

\theoremstyle{thmstyleone}%
\newtheorem{theorem}{Theorem}
\newtheorem{lemma}{Lemma}%

\theoremstyle{thmstyletwo}%
\newtheorem{example}{Example}%

\theoremstyle{thmstylethree}%
\newtheorem{definition}{Definition}%

\begin{document}

\title{Flat subspaces of the $SL(n,\mathbb{R})$ chiral equations}


\author*[1]{\fnm{I. A.} \sur{Sarmiento-Alvarado}}\email{ignacio.sarmiento@cinvestav.mx}

\author[2]{\fnm{P.} \sur{Wiederhold}}\email{petra.wiederhold@cinvestav.mx}
\equalcont{These authors contributed equally to this work.}

\author[1]{\fnm{T.} \sur{Matos}}\email{tonatiuh.matos@cinvestav.mx}
\equalcont{These authors contributed equally to this work.}

\affil*[1]{\orgdiv{Departamento de F\'{\i}sica}, \orgname{Centro de Investigaci\'on y de Estudios Avanzados del Instituto Polit\'ecnico Nacional}, \orgaddress{\street{Av. Instituto Polit\'ecnico Nacional 2508}, \city{San Pedro Zacatenco}, \postcode{07360}, \state{CDMX}, \country{M\'exico}}}

\affil[2]{\orgdiv{Departamento de Control Autom\'atico}, \orgname{Centro de Investigaci\'on y de Estudios Avanzados del Instituto Polit\'ecnico Nacional}, \orgaddress{\street{Av. Instituto Polit\'ecnico Nacional 2508}, \city{San Pedro Zacatenco}, \postcode{07360}, \state{CDMX}, \country{M\'exico}}}

\abstract{
In this work, we introduce a method for finding exact solutions to the vacuum Einstein field equations in higher dimensions from a given solution to the chiral equation.
When considering a $n + 2$-dimensional spacetime with $n$ commutative Killing vectors, the metric tensor can take the form $\hat g = f ( \rho, \zeta ) ( d \rho^2 + d \zeta^2 ) + g_{\mu \nu} ( \rho, \zeta ) d x^\mu d x^\nu$.
Then, the Einstein field equations in vacuum reduce to a chiral equation, $( \rho g_{, z} g ^{-1} )_{, \bar z} + ( \rho g_{, \bar z} g ^{-1} )_{, z} = 0$, and two differential equations, $( \ln f \rho ^{1-1/n} )_{, Z} = \frac{\rho}{2} \operatorname{tr} ( g_{, _Z} g^{-1} )^2$, where $g \in SL( n, \mathbb{R} )$ is the normalized matrix representation of $g_{\mu \nu}$, $z = \rho + i \zeta$ and $Z = z, \bar z$.
We use the ansatz $g = g ( \xi^a )$, where the parameters $\xi^a$ depend on $z$ and $\bar z$ and satisfy a generalized Laplace equation, $( \rho \xi^a _{, z} )_{, \bar z} + ( \rho \xi^a _{, \bar z} )_{, z} = 0$.
The chiral equation to the Killing equation, $A_{a , \xi^b} + A_{b , \xi^a} = 0$, where $A_a = g_{, \xi^a} g^{-1}$.
Furthermore, we assume that the matrices $A_a$ commute with each other; in this way, they fulfill the Killing equation.
}


\maketitle

\section{Introduction}

It is well known that the Kaluza-Klein theory, based on a five-dimensional Riemannian geometry, is a proposal to unify gravity and electromagnetism \cite{kaluza_921, Klein:1926tv}.
Its impact was such that it was generalized to dimensions more than five, with the goal of unifying gravitation with electroweak and strong interactions \cite{Bailin:1987jd}.
In a $d$-dimensional theory, it is assumed that the whole space $U$ has a topology of a principal fiber bundle \cite{nakahara_2003}, where the base space is our 4-dimensional spacetime and the fiber is a Lie group $G$, which is known as the inner space.
In addition, $U$ is supposed to have a metric $\hat{g}$, which is invariant under the left action of $G$ over $U$.
In local coordinates, $\hat{g}$ has the form
\begin{equation}
\label{d-dim Kaluza Klein metric}
    \hat{g} = g_{\alpha \beta} dx^\alpha dx^\beta + h_{m n} ( \omega^m + A^m_\alpha dx^\alpha ) ( \omega^n + A^n_\beta dx^\beta ) ,
\end{equation}
where $g_{\alpha \beta}$, $h_{m n}$ and $A^m_\alpha$ depend on $x^\alpha$; $\{ \omega^m \}$ is a basis of left-invariant 1-forms on $G$ for all $\alpha, \beta \in \{ 1, \ldots, 4 \}$ and $m, n \in \{ 5, \ldots, d \}$.
In Eq. (\ref{d-dim Kaluza Klein metric}), $g_{\alpha \beta} dx^\alpha dx^\beta$ is the metric on spacetime, $h_{m n} \omega^m \omega^n$ is the metric on the fiber, and $A^m_\alpha dx^\alpha$ is the $G$-connection.
The Kaluza-Klein theory is one of the first theories based on higher dimensions and is considered as an important precursor to the string theory.
The string theory and the superstring theory unify all interactions in 26 and 10 dimensions, respectively.
Therefore, knowing exact solutions to the Einstein field equations (EFE) in higher dimensions gives a better understanding of gravity in relation to the other fundamental interactions of nature and matter.

Finding exact solutions to the EFE is not an easy work, especially in higher-dimensional spaces, due to the non-linearity and complexity of the field equations.
For a $(n + 2)$-dimensional space endowed with a metric, which admits $n$ commutative Killing vectors, there exists a coordinate system such that the metric tensor has the form \cite{Matos89,Sarmiento-Alvarado2023}
\begin{equation}
\label{metric tensor}
    \hat g = f ( \rho, \zeta ) ( d \rho^2 + d \zeta^2 ) + g_{i j} ( \rho, \zeta ) d x^i d x^j \text{ for all } i, j \in \{ 3, \ldots, n+2 \} \text{,}
\end{equation}
where $f$ is a positive function and $\rho = \sqrt{ -\det g_{i j} }$.
Thus, the EFE in vacuum, $R_{A B} = 0$ for all $A, B \in \{ 1, \ldots, n+2 \}$,  are separated into parts:
\begin{align}
\label{chiral eq g}
&    ( \rho g_{, z} g ^{-1} )_{, \bar z} + ( \rho g_{, \bar z} g ^{-1} )_{, z} = 0
,
\\\label{SL invariant field eq f}
&
    ( \ln f \rho ^{1-1/n} )_{, Z} = \frac{1}{2} \rho \operatorname{tr} ( g_{, _Z} g^{-1} )^2
    \text{ for } Z = z \text{ and } Z = \bar z,
\end{align}
where $g$ is a symmetric matrix in $SL ( n, \mathbb{R} )$, defined in terms of the components of the metric tensor as $g = -\rho^{-2/n} g_{i j}$, and $z = \rho + i \zeta$ .

Several exact solutions to the chiral equation \eqref{chiral eq g} can be found in the literature.
In \cite{10.1063/1.529991,Matos1993}, the authors solved the chiral equation (\ref{chiral eq g}) considering $g$ as a member of the groups $SL ( 2, \mathbb{R} )$, $SL ( 3, \mathbb{R} )$ and $SL ( 4, \mathbb{R} )$.
Recently, the authors in \cite{Sarmiento-Alvarado2023} solved it for any $n>1$.
To achieve this, they used the properties of Jordan matrices and algebraic methods from linear algebra theory.

The advantage of knowing the exact solutions of the $SL ( n, \mathbb{R} )$-invariant chiral equation is that it allows the construction of solutions to the EFE in vacuum.
For this reason, in this work, we solve the chiral equation for a symmetric matrix $g$ in $SL ( n, \mathbb{R} )$ using the algebraic method introduced in \cite{Sarmiento-Alvarado2023}, as well as techniques from linear algebra theory.
Once $g$ is obtained, we present the procedure for constructing exact solutions to vacuum EFE written in the form of a Kaluza-Klein metric \eqref{d-dim Kaluza Klein metric}.

This article is organized as follows.
In Section \ref{section: chiral equation}, we solve the $SL ( n, \mathbb{R} )$-invariant chiral equation for a symmetric matrix.
Given a solution to the chiral equation, in Section \ref{section: Kaluza-Klein metrics}, we show the procedure to obtain an exact solution to the EFE written as a Kaluza-Klein metric \eqref{d-dim Kaluza Klein metric}.
In Section \ref{section: centralizer}, we determined a pair of commuting matrices using the centralizer of matrices.
In Section \ref{section: commutative algebras}, we find the largest commutative algebras contained in $\mathfrak{sl} ( n, \mathbb{R} )$.
In Section \ref{section: subspace ia}, we introduce the set $\mathcal{I} ( \mathfrak{A} )$.
In Section \ref{section: example}, we apply our results to find exact solutions to the EFE in vacuum.
Finally, in Section \ref{section: conclusions}, we present some conclusions.

Throughout this paper, the indices $a, b, c$ take values in $\{ 1, \ldots, r \}$ for some $r \in \{ 1, \ldots, n - 1 \}$.
$\mathbf{M}_{m , n}$ is the set of all $m \times n$ matrices over $\mathbb{R}$.
In the case $m = n$, we only write $\mathbf{M}_m$.
$I_m$ is the identity matrix.
The zero matrix is denoted by $0_{m,n}$ or $0_m$ if $m = n$, or even by $0$ when there is no confusion.
We denote the set of all  $m \times m$ matrices that satisfy $g^T = g$ by $\mathbf{Sym}_m$.
$\mathbb{R} [ x ]$ is the set of all polynomials over $\mathbb{R}$.

\section{Chiral equation}
\label{section: chiral equation}

In this section, we will solve the chiral equation.

We assume that $g$ depends on a set of parameters $\xi^a = \xi^a ( z, \bar{z} )$, that is $g = g ( \xi^a ( z, \bar{z} ) )$.
Furthermore, the parameters $\xi^a$ fulfill the geodesic equation
\begin{equation}
\label{geodesic eq}
    \left( \rho \xi^c _{, z} \right)_{, \bar z} + \left( \rho \xi^c _{, \bar z} \right)_{, z} + 2 \rho \Gamma^c_{a b} \xi^a_{, z} \xi^b_{, \bar z} = 0
    \text{ for all } a, b, c \in \{ 1, \ldots, r \}
\end{equation}
of a Riemannian space $V$ with Christoffel symbols $\Gamma^c_{a b}$.
Here, $1 < r \leq n$.
Thus, the chiral equation (\ref{chiral eq g}) becomes
\begin{equation}
\label{Killing eq A Riemannian space}
    A_{a ; b} + A_{b ; a} = 0
\text{,}
\end{equation}
where we have defined the matrices
\begin{equation}
\label{def matrices A}
    A_a = g_{, \xi^a} g^{-1}
\text{.}
\end{equation}
For a flat space $V$, the Christoffel symbols are zero, thus Eqs. \eqref{geodesic eq} and \eqref{Killing eq A Riemannian space} reduces to
\begin{align}
\label{gen Laplace eq}
    \left( \rho \xi^c _{, z} \right)_{, \bar z} + \left( \rho \xi^c _{, \bar z} \right)_{, z} & = 0 \text{,}
\\\label{Killing eq A flat space}
    A_{a , b} + A_{b , a} & = 0 \text{,}
\end{align}
respectively.

From Eq. \eqref{def matrices A}, we obtain that the matrices $A_a$ are traceless as a consequence of the fact that the determinant of $g$ is constant.
This means that the matrices $A_a$ are members of the Lie algebra $\mathfrak{sl} ( n, \mathbb{R} )$.
Observe that the chiral equation (\ref{chiral eq g}) is invariant under transformations $g \to C g C ^T$, then the matrices $A_a$ vary as $A_a \to C A_a C^{-1}$, where $C$ is a constant matrix in $SL ( n, \mathbb{R} )$.
In this way, the set of matrices $A_a$ is separated into similarity equivalence classes.

The matrices $A_a$ satisfy Eq. \eqref{Killing eq A flat space} when they are constants.
Since the partial derivative of $A_a$ with respect to $\xi^b$, denoted by $A_{a, b}$, is
\begin{equation}
    A_{a, b} = -\frac{1}{2} [A_a, A_b]\text{,}
\end{equation}
then $\{ A_1, \ldots, A_r \}$ is a set of pairwise commuting matrices in $\mathfrak{sl} ( n, \mathbb{R} )$.
The matrix $g ( \xi^a )$ is determined by Lemma \ref{lemma: solution diff eq g}.
\begin{lemma}
\label{lemma: solution diff eq g}
    Let $\{ A_1, \ldots, A_r \}$ be a subset of pairwise commuting matrices in $\mathbf{M}_n$ and let $g \in \mathbf{Sym}_{n}$ be a matrix function of the variables $\xi^1, \ldots, \xi^r$.
    If $g_{, \xi^a} = A_a g$ for all $a \in \{ 1, \ldots, r \}$, then $g ( \xi^a ) = e ^{\xi^a A_a} g_0$, where $g_0$ is a constant matrix that satisfies $A_a g_0 = g_0 A_a ^T$.
\end{lemma}
\begin{proof}
    First, we solve $g_{, \xi^1} = A_1 g$.
    Then, $g = e ^{\xi ^1 A_1} h$, where $h$ is a matrix function that depends on $\xi_2, \ldots, \xi_r$.
    Substituting this solution into $g_{, \xi^2} = A_2 g$, we obtain $h_{, \xi^2} = e ^{-\xi ^1 A_1} A_2 e ^{\xi ^1 A_1} h$.
    Since $A_1$ and $A_s$ commute, it follows that $h_{, \xi^2} = A_2 h$. 
    Following the same steps, we get $ g = e ^{ \xi^1 A_1 } \cdots e ^{ \xi^r A_r } g_0 = e ^{\xi^a A_a} g_0$, where $g_0$ is a constant matrix.

    From $g_{, \xi^a } = ( g_{, \xi^a } )^T$, we get the intertwining relations $A_a g = g A_a ^T$.
    Then, $A_a e ^{\xi^a A_a} g_0 = e ^{\xi^a A_a} g_0 A_a ^T$, so that $A_a g_0 = g_0 A_a ^T$ for all $a \in \{ 1, \ldots, r \}$.
\end{proof}

Once we know the matrix $g ( \xi^a )$, the next step is to find a set $\{ \xi^a ( z, \bar{z} ) \}$ of solutions to the generalized Laplace equation \eqref{gen Laplace eq}.
Suppose that we know such a set.
We then compute the exponential of the linear combination $\xi^a ( z, \bar{z} ) A_a$ to build a solution $g ( z, \bar{z} )$ to the chiral equation \eqref{chiral eq g}, which is given by $g ( z, \bar{z} ) = e^{ \xi^a ( z, \bar{z} ) A_a } g_0$.

\section{Kaluza-Klein metrics}
\label{section: Kaluza-Klein metrics}

In this section, we will show how to determine a metric tensor given a solution of the chiral equation.

In order to obtain a metric tensor written in the form \eqref{d-dim Kaluza Klein metric} from a given solution $g$ of Eq. \eqref{chiral eq g}, we partition $g$ as
\begin{equation}
    g
    = -\left[\begin{array}{cc}
        P & Q \\
        Q^T & H \\ 
    \end{array}\right] \text{,}
\end{equation}
respectively, where $P \in \mathbf{Sym}_2$, $Q \in \mathbf{M}_{2, n - 2}$ and $H \in \mathbf{Sym}_{n - 2}$.
If $H$ is a non-singular matrix, we use Aitken block diagonalization formula for $g$ \cite{Horn2005} to obtain
\begin{equation}
\label{decomposition g}
    g = - \left[\begin{array}{cc}
        I_2 & Q H^{-1} \\
        & I_{n - 2} \\ 
    \end{array}\right] \left[\begin{array}{cc}
        g/H & \\
        & H \\ 
    \end{array}\right] \left[\begin{array}{cc}
        I_2 & Q H^{-1} \\
        & I_{n - 2} \\ 
    \end{array}\right]^T ,
\end{equation}
where $g/H = P - Q H^{-1} Q^T \in \mathbf{Sym}_2$ is the Schur complement of $H$ in $g$ \cite{Horn2005}.
Therefore, an exact solution to EFE is given by the metric tensor
\begin{equation}
\label{sol metric tensor}
    \hat g = f ( d\rho^2 + d\zeta^2 ) + g_{\mu \nu } d x^\mu d x^\nu + h_{i j} ( dx^i + A_\mu ^i dx^\mu ) ( dx^j + A_\nu ^j dx^\nu ) \text{,}
\end{equation}
whose components are $g_{\mu \nu } = \rho^\frac{2}{n} g/H$, $h_{i j} = \rho^\frac{2}{n} H$ and $A_\mu^i =  H^{-1} Q^T$ for all $\mu, \nu \in \{ 3, 4 \}$ and $i, j \in \{ 5, \ldots, n + 2 \}$.
Recall that $f$ is obtained by solving Eq. \eqref{SL invariant field eq f} for the set $\{ \xi^a ( z, \bar{z} ) \}$.

From Schur determinant formula for $g$ \cite{Horn2005}, we have $\det g/H \det H = -1$, so $g/H$ is a non-singular matrix.
This means that the inverse of $\hat g$ can be expressed in terms of the inverses of $f$, $g/H$ and $H$.
To compute $g^{-1}$, we use the Banachiewicz inversion formula for $g$ \cite{Horn2005}, obtaining
\begin{equation}
    g^{-1}
    = -\left[\begin{array}{cc}
        ( P / H )^{-1} & -( P / H )^{-1} Q H^{-1} \\
        -( ( P / H )^{-1} Q H^{-1} )^T & H^{-1} \\
    \end{array}\right].
\end{equation}
Thus, $g^{\rho \rho} = g^{\zeta \zeta} = f^{-1}$, $g^{\mu \nu} = \rho^{ -\frac{2}{n} } ( P / H )^{-1}$, $g^{\mu j} = - \rho^{ -\frac{2}{n} }  ( P / H )^{-1} Q H^{-1}$ and $g^{i j} = \rho^{ -\frac{2}{n} } ( H^{-1} - H^{-1} Q^T ( P / H )^{-1} Q H^{-1} )$ for all $\mu, \nu \in \{ 3, 4 \}$ and $i, j \in \{ 5, \ldots, n + 2 \}$.

\section{The centralizer of a set of matrices}  %
\label{section: centralizer}

In this section, we will introduce the set of all matrices that commute with a given matrix or a set of matrices.
For more information on pairwise commuting matrices, see \cite{horn_johnson_1985, Gantmacher59, DOLINAR20132904, FERRER20133945}.

\begin{definition}
    For any non-empty set $\mathfrak{A} \subset \mathbf{M}_{n}$, define its \textbf{centralizer} as the set
    \begin{equation}
        \mathcal{C} ( \mathfrak{A} ) = \left\{
            B \in \mathbf{M}_{n} :
            A B = B A
            \quad
            \forall A \in \mathfrak{A}
        \right\}   \  .
    \end{equation}
\end{definition}
When $\mathfrak{A} = \{ A \} $ is a singleton set, we write $\mathcal{C} ( A )$ instead of $\mathcal{C} ( \{ A \} )$.
For a finite set $\{ A_1, \ldots, A_r \}$, we write $\mathcal{C} ( A_1, \ldots, A_r )$ or $\mathcal{C} \{ A_a \}$.

\begin{theorem}
    For any non-empty $\mathfrak{A} \subset \mathbf{M}_{n}$, $\mathcal{C} ( \mathfrak{A} )$ is a subspace of $\mathbf{M}_{n}$.
\end{theorem}

\begin{theorem}
\label{theo: centralizer union sets}
    For any non-empty subsets $\mathfrak{A}, \mathfrak{B} $ of $ \mathbf{M}_{n}$, 
    \begin{equation}
        \mathcal{C} ( \mathfrak{A} \cup \mathfrak{B} ) = \mathcal{C} ( \mathfrak{A} ) \cap \mathcal{C} ( \mathfrak{B} )
    \text{.}
    \end{equation}
\end{theorem}

\begin{lemma}
\label{lemma: subset centralizer}
   For any non-empty subsets $\mathfrak{A}, \mathfrak{B} $ of $ \mathbf{M}_{n}$, $\mathfrak{A} \subset \mathfrak{B}$ implies $\mathcal{C} ( \mathfrak{B} ) \subset \mathcal{C} ( \mathfrak{A} )$.
\end{lemma}

\begin{lemma}
\label{lemma: centralizer of pair commuting matrices set}
    If $\mathfrak{A} \subset \mathbf{M}_{n}$ is a set of pairwise commuting matrices then $\mathfrak{A} \subset \mathcal{C} ( \mathfrak{A} )$.
\end{lemma}
\begin{proof}
    If $X \in \mathfrak{A}$, then $X A = A X$ for all $A \in \mathfrak{A}$, so that $X \in \mathcal{C} ( \mathfrak{A} )$.
    Hence, $\mathfrak{A} \subset \mathcal{C} ( \mathfrak{A} )$.
\end{proof}

In \cite{10.1063/1.529991}, the authors worked with $\mathfrak{sl} ( 3, \mathbb{R} )$ and determined $g$ considering a flat subspace spanned by two matrices $\sigma_1$ and $\sigma_3$ which commute with each other.
To find the matrices $\sigma_1$ and $\sigma_3$, they assumed that $\sigma_1$ is a representative of some similarity equivalence class of $\mathfrak{sl} ( 3, \mathbb{R} )$.
Then, they solved the linear matrix equation $[ \sigma_1, \sigma_3 ] = 0$ to find $\sigma_3$.

In order to determine a pair $A_1$ and $A_2$ of commuting matrices, we suppose that $A_1$ is representative of some equivalence class of $\mathfrak{sl} ( n, \mathbb{R} )$.
Since $A_2$ commutes with $A_1$, $A_2$ belongs to the centralizer of $A_1$.

In a previous work \cite{Sarmiento-Alvarado2023}, we classified the equivalence classes of $\mathfrak{sl} ( n, \mathbb{R} )$ into five types according to the type of eigenvalues.
This means that there are five types of matrices $A_1$.
In the following, we present the five types of equivalence classes of $\mathfrak{sl} ( n, \mathbb{R} )$ in their real Jordan form.

\begin{definition}
\label{def: jordan block 1st lambda}
    A real Jordan block of the first kind $ J _m (\lambda) \in \mathbf{M}_{m} $ with a real eigenvalue $\lambda$ is defined as a matrix of the form
    \begin{equation}
    \label{jordan block 1st real}
        J _m (\lambda) = \left[\begin{array}{cccc}
            \lambda & 1 & \cdots & 0 \\
            & \lambda & \cdots & 0 \\
            && \ddots & \vdots \\
            &&& \lambda \\
        \end{array}\right].
    \end{equation}
\end{definition}

\begin{definition}
\label{def: jordan block 2do lambda}
    Let $ m _1, \ldots, m _p, m$ be positive integers such that $ m = m _1 + \ldots + m _p $.
    A real Jordan block of the second kind $ J _{ m _1, \ldots, m _p } ( \lambda ) \in \mathbf{M}_{m} $ with real eigenvalue $\lambda$ is defined as a block matrix of the form
    \begin{equation}
    \label{jordan block 2nd real}
        J _{ m _1, \ldots, m _p } ( \lambda ) = \operatorname{diag} [
            J _{m _1} ( \lambda ),
            \ldots,
            J _{m _p} ( \lambda )
        ],
    \end{equation}
    where $J _{m _i} ( \lambda )$ are as in Definition \ref{def: jordan block 1st lambda} for all $i \in \{ 1, \ldots, p \}$.
\end{definition}

\begin{definition}
\label{def: jordan block 1st Lambda}
    A real Jordan block of the first kind $ J ^n ( \Lambda ) \in \mathbf{M}_{2 n} $ with a pair of complex conjugate eigenvalues
    \begin{equation}
        \Lambda = \left[\begin{array}{cc}
        \alpha  &   -\beta  \\
        \beta   &   \alpha
    \end{array}\right] \in \mathbf{M}_{2}
    \text{, with } 0 < \beta
    \text{,}
    \end{equation}
    is given as a block matrix of the form
    \begin{equation}
    \label{jordan block 1st complex conjugates}
        J ^n ( \Lambda ) = \left[\begin{array}{cccc}
            \Lambda & I_2 & \cdots & 0_2 \\
            & \Lambda & \cdots & 0_2 \\
            && \ddots &   \vdots  \\
            &&& \Lambda \\
        \end{array}\right].
    \end{equation}
\end{definition}

\begin{definition}
\label{def: jordan block 2do Lambda}
    Let $ n _1, \ldots, n _q , n$ be positive integers such that $ n = n _1 + \ldots + n _q $.
    A real Jordan block of the second kind $ J ^{ n _1, \ldots, n _q } ( \Lambda ) \in \mathbf{M}_{2 n} $ with a pair of complex conjugate eigenvalues
    \begin{equation}
        \Lambda = \left[\begin{array}{cc}
        \alpha  &   -\beta  \\
        \beta   &   \alpha
    \end{array}\right] \in \mathbf{M}_{2}
    \text{, with } 0 < \beta,
    \end{equation}
    is defined as a block diagonal matrix
    \begin{equation}
    \label{jordan block 2nd complex conjugates}
        J ^{ n _1, \ldots, n _q } ( \Lambda ) =
        \operatorname{diag} (
            J ^{n _1} ( \Lambda ),
            \ldots,
            J ^{n _q} ( \Lambda )        
        ),
    \end{equation}
    where $J ^{n _i} ( \Lambda ) $ are given as in Definition \ref{def: jordan block 1st Lambda} for all $i \in \{ 1, \ldots, q \}$.
\end{definition}

\begin{definition}
\label{def: jordan matrix}
    Let $i \in \{ 1, \ldots, p \}$, $k \in \{ 1, \ldots, q \}$, and 
    $m, m ^i, m ^i _1, \ldots, m ^i _{r _i}, n, n ^k, n ^k _1, \ldots, n ^k _{s _k}$ be positive integers such that $ m ^i = m ^i _1 + \ldots + m ^i _{r _i} $, $ n ^k = n ^k _1 + \ldots + n ^k _{s _k} $, $ m = m ^1 + \ldots + m ^p $ and $ n = n ^1 + \ldots + n ^q $.
    A Jordan matrix $ J \in \mathbf{M}_{m + 2 n} $ with distinct real eigenvalues $\lambda_i$ and different complex conjugate eigenvalues
    \begin{equation}
        \Lambda_k = \left[\begin{array}{cc}
        \alpha_k  &   -\beta_k  \\
        \beta_k   &   \alpha_k
    \end{array}\right] \in \mathbf{M}_{2}
    \text{, with } 0 < \beta_k
    \text{,}
    \end{equation}
    is defined as a block diagonal matrix of the form
    \begin{equation}
        J = \operatorname{diag} [
            J _{m ^1 _1, \ldots, m ^1 _{r _1}} ( \lambda _1 ),
            \ldots,
            J _{m ^p _1, \ldots, m ^p _{r _p}} ( \lambda _p ),
            J ^{n ^1 _1, \ldots, n ^1 _{s _1}} ( \Lambda _1 ),
            \ldots,
            J ^{n ^q _1, \ldots, n ^q _{s _q}} ( \Lambda _q )        
        ],
    \end{equation}
    where $ J _{m ^i _1, \ldots, m ^i _{r _i}} ( \lambda _i ) $ and $ J ^{n ^k _1, \ldots, n ^k _{s _k}} ( \Lambda _k ) $ are from Definitions \ref{def: jordan block 2do lambda} and \ref{def: jordan block 2do Lambda}, respectively.
\end{definition}

In order to determine the matrix $A_2$, we compute the centralizers of the five types of equivalence classes of $\mathfrak{sl} ( n, \mathbb{R} )$.
Some results on the centralizers of Jordan matrices were obtained from \cite{FERRER20133945}.
The authors of \cite{FERRER20133945} use the transpose of our Jordan block as their Jordan block.

\begin{theorem}
\label{theorem: centralizer jordan block 1st lambda}
    If $J _m ( \lambda )$ is the Jordan matrix defined in Definition \ref{def: jordan block 1st lambda}, then its centralizer, denoted by $\mathcal{C} ( J _m ( \lambda ) )$, is given as
    \begin{equation}
    \mathcal{C} ( J _m ( \lambda ) )
    = \left\{ A\in \mathbf{M}_{m} :\  A = \left[\begin{array}{cccc}
    
        x _1    &   x _2    &   \cdots  &   x _m    \\
        0       &   x _1    &   \cdots  &   x _{m-1}    \\
        \vdots  &   \vdots  &   \ddots  &   \vdots  \\
        0       &   0       &   \cdots  &   x _1
    \end{array}\right]
    \right\}.
\end{equation}
\end{theorem}

\begin{theorem}
\label{theorem: centralizer jordan block 2do lambda}
    Let $J _{ m _1, \ldots, m _p } ( \lambda )$ be the matrix from Definition \ref{def: jordan block 2do lambda}.
    Every matrix $X \in \mathcal{C} ( J _{ m _1, \ldots, m _p } ( \lambda ) )$ is of the form
\begin{equation}
    X = \left[\begin{array}{ccc}
        X _{1 1}    &   \cdots  &   X _{1 p}    \\
        \vdots      &   \ddots  &   \vdots  \\
        X _{p 1}    &   \cdots  &   X _{p p}
    \end{array}\right],
\end{equation}
where
\begin{itemize}
    \item $ X _{i j} \in \mathcal{C} ( J _{m _i} ( \lambda ) ) $ if $ m _i = m _j $,
    \item $ X _{i j} = \left[\begin{array}{cc}
            0 & Y _{ij}
        \end{array}\right] $ with $ Y _{i j} \in \mathcal{C} ( J _{m _i} ( \lambda ) ) $ if $ m _i < m _j $,
    \item $ X _{i j} = \left[\begin{array}{c}
            Y _{i j} \\ 0
        \end{array}\right] $ with $ Y _{i j} \in \mathcal{C} ( J _{m _j} ( \lambda ) ) $ if $ m _i > m _j $,
\end{itemize}
for all $i, j \in \{ 1, \ldots, p \}$.
\end{theorem}


\begin{lemma}
\label{lemma: centralizer jordan block 1st Lambda}
    Let
    \begin{equation}
        \Lambda = \left[\begin{array}{cc}
        \alpha  &   -\beta  \\
        \beta   &   \alpha
    \end{array}\right] \in \mathbf{M}_{2}.
    \end{equation}
    Then,
    \begin{equation}
    \mathcal{C} ( \Lambda ) = \left\{  A\in \mathbf{M}_{2} :
        A= \left[\begin{array}{cc}
            x   &   -y  \\
            y   &   x
        \end{array}\right] 
    \right\}.
    \end{equation}
\end{lemma}
\begin{proof}
    Let $ \mathfrak{X} = \left[\begin{array}{cc} x & z \\ y & t \end{array}\right] \in \mathbf{M}_{2} $.
    The intertwining relation $ \Lambda \mathfrak{X} = \mathfrak{X} \Lambda $ implies $ t = x $ and $ z = - y $.
\end{proof}

\begin{lemma}
\label{lemma: linear eq centralizer Lambda}
    Let $\mathcal{C} ( \Lambda )$ be the centralizer of
    \begin{equation}
        \Lambda = \left[\begin{array}{cc}
        \alpha  &   -\beta  \\
        \beta   &   \alpha
    \end{array}\right] \in \mathbf{M}_2.
    \end{equation}
    If $ X \in \mathbf{M}_{2} $ and $ Y \in \mathcal{C} ( \Lambda )$, then
    \begin{equation}
        \Lambda X = X \Lambda + Y \iff X \in \mathcal{C} ( \Lambda ), Y = 0
    \text{.}
    \end{equation}
\end{lemma}
\begin{proof}
    Let $ X = \left[\begin{array}{cc} x & z \\ y & t \end{array}\right] \in \mathbf{M}_{2} $.
    Each $ Y \in \mathcal{C} (  \Lambda ) $ has the form $ Y = \left[\begin{array}{cc} \gamma & - \delta \\ \delta & \gamma \end{array}\right] \in \mathbf{M}_{2} $.
    The relation $ \Lambda X = X \Lambda + Y $ implies $ \gamma = \delta = 0 $, $ t = x $ and $ z = -y $.
    Now, if $ X \in \mathcal{C} (  \Lambda ) $ then $ \Lambda X = X \Lambda $, hence $ Y = 0 $.
\end{proof}

\begin{lemma}
\label{lemma: J X = X Lambda + Y}
    Let $ J ^n ( \Lambda ) $ be a Jordan block of the first kind due to Definition \ref{def: jordan block 1st Lambda}, and
    \begin{equation}
        X ^T = \left[\begin{array}{ccc} X_1 ^T & \cdots & X_n ^T \end{array}\right], Y ^T = \left[\begin{array}{ccc} Y_1 ^T &\cdots & Y_n ^T \end{array}\right] \in \mathbf{M}_{2, 2 n}
    \end{equation}
    where $X_i, Y_i \in \mathbf{M}_{2}$ for all $i \in \{ 1, \ldots, n \}$.
    If $ J ^n ( \Lambda ) X = X \Lambda + Y $, then $X_i \in \mathcal{C} ( \Lambda )$ for all $i \in \{ 1, \ldots, n \}$ and
    \begin{equation}
        Y ^T = \left[\begin{array}{cccc} X_2 ^T & \ldots & X_n ^T & 0_2 \end{array}\right].   
    \end{equation}
\end{lemma}
\begin{proof}
    The intertwining relation $ J ^n ( \Lambda ) X = X \Lambda + Y $ implies the equations
    \begin{equation}
    \begin{aligned}
        \Lambda X_n
    &   = X_n \Lambda + Y_n
    \\  \Lambda X_{n-1} + X_n
    &   = X_{n-1} \Lambda + Y_{n-1}
    \\& \vdots
    \\  \Lambda X_1 + X_2
    &   = X_1 \Lambda + Y_1
    \end{aligned}
    \end{equation}
    Applying Lemma \ref{lemma: linear eq centralizer Lambda}, we get $ X_i \in \mathcal{C} ( \Lambda ) $ for all $i \in \{ 1, \ldots, n \}$ and $ X_2 = Y_1, \ldots, X_n = Y_{n-1}, Y_n = 0 $.
\end{proof}

\begin{theorem}
\label{theorem: centralizer jordan block 1st Lambda}
    The centralizer of $J ^n ( \Lambda ) $, defined in Definition \ref{def: jordan block 1st Lambda}, is
    \begin{equation}
    \mathcal{C} ( J ^n ( \Lambda ) )
    = \left\{\left[\begin{array}{cccc}
        X _1    &   X _2    &   \cdots  &   X _n    \\
        &   X _1    &   \cdots  &   X _{n-1}    \\
        &&   \ddots  &   \vdots  \\
        &&&   X _1
    \end{array}\right] \in \mathbf{M}_{2 n}
    :\   X_1, \ldots, X_n \in \mathcal{C} ( \Lambda )
    \right\}
\end{equation}
\end{theorem}
\begin{proof}
    Let $ i \in \{ 1, \ldots, n \} $.
    If $ X \in \mathcal{C} ( J ^n ( \Lambda ) ) $ then $ J ^n ( \Lambda ) X = X J ^n ( \Lambda ) $.
    If the matrix $X$ is partitioned as $ \left[\begin{array}{ccc} X_1 & \cdots & X_n \end{array}\right] $ with $ X_i \in \mathbf{M}_{2 n, 2} $, then
    \begin{equation}
    \begin{aligned}        
        J ^n ( \Lambda ) X_1
    &   = X_1 \Lambda
    \\  J ^n ( \Lambda ) X_2
    &   = X_2 \Lambda + X_1
    \\& \vdots
    \\  J ^n ( \Lambda ) X_n
    &   = X_n \Lambda + X_{n - 1}
    \end{aligned}
    \end{equation}
    Now, we define $X_n ^T = \left[\begin{array}{ccc} Y_1 ^T & \cdots & Y_n ^T \end{array}\right]$, where $Y_i \in \mathbf{M}_{2}$.
    Using Lemma \ref{lemma: J X = X Lambda + Y}, we find $Y_i \in \mathcal{C} (  \Lambda )$ and
    \begin{equation}
        X_{n - 1} ^T = \left[\begin{array}{cccc} Y_2 ^T & \cdots & Y_n ^T & 0_2 \end{array}\right],
        \cdots ,
        X_1^T = \left[\begin{array}{cccc} Y_n ^T & 0_2 & \cdots & 0_2 \end{array}\right]       
    \end{equation}
\end{proof}

\begin{lemma}
\label{lemma: Lambda centralizer Jm < Jn}
    Let $m$ and $n$ be two positive integers such that $m < n$.
    Let $J ^m ( \Lambda )$ and $J ^n ( \Lambda )$ be Jordan matrices defined in Definition \ref{def: jordan block 1st Lambda}.
    Let $X \in \mathbf{M}_{2m, 2n}$.
    \begin{equation}
        J ^m ( \Lambda ) X = X J ^n ( \Lambda ) \iff X = \left[\begin{array}{cc} 0 & Y \end{array}\right], Y \in \mathcal{C} ( J ^m ( \Lambda ) )
    \end{equation}
\end{lemma}
\begin{proof}
    Let $i \in \{ 1, \ldots, m \} $ and $k \in \{ 1, \ldots, n \} $.
    Any $ X \in \mathbf{M}_{2m, 2n} $ can be partitioned as $ \left[\begin{array}{ccc} X_1 & \cdots & X_n \end{array}\right] $, where $ X_k \in \mathbf{M}_{2 m, 2}$.
    Let $X_n ^T = \left[\begin{array}{ccc} Y_1^T & \cdots & Y_m^T \end{array}\right]$ with $Y_i \in \mathbf{M}_{2}$.
    If $ J ^m ( \Lambda ) X = X J ^n ( \Lambda ) $, then
    \begin{equation}
    \begin{aligned}
        J ^m ( \Lambda ) X_1
    &   = X_1 \Lambda
    \\  J ^m ( \Lambda ) X_2
    &   = X_2 \Lambda + X_1
    \\& \vdots
    \\  J ^m ( \Lambda ) X_n
    &   = X_n \Lambda + X_{n - 1}        
    \end{aligned}
    \end{equation}
    Using Lemma \ref{lemma: J X = X Lambda + Y} we find $Y_i \in \mathcal{C} ( \Lambda )$ and
    \begin{align}
        X_{n - 1} ^T = \left[\begin{array}{cccc} Y_2^T & \cdots & Y_m^T & 0_2 \end{array}\right], \cdots, X_{n - m + 1} ^T = \left[\begin{array}{cccc} Y_m^T & 0_2 & \cdots & 0_2 ]\end{array}\right]
    &&
        X_1 = \cdots = X_{n - m} = 0_{2 m, 2}
    \end{align}
    Thus, we can write $ X = \left[\begin{array}{cc} 0 & Y \end{array}\right] $ with $ Y \in \mathcal{C} ( J ^m ( \Lambda ) ) $.
    
    Now, let us partition $ J ^n ( \Lambda ) $ as 
    \begin{equation}
    \begin{aligned}
        J ^n ( \Lambda )
    &   = \left[\begin{array}{cc}
            J ^{n - m} ( \Lambda ) & E_{n - m,1}
        \\  0 & J ^m ( \Lambda )
        \end{array}\right],
    &   \text{where }
        E_{n-m,1}
    &   = \left[\begin{array}{cccc}
            0_2 &   0_2 &   \cdots  &   0_2
        \\  \vdots  &   \vdots  &   \ddots  &   \vdots
        \\  I_2 &   0_2 &   \cdots  &   0_2
        \end{array}\right].
    \end{aligned}
    \end{equation}
    If $ Y \in \mathcal{C} ( J ^m ( \Lambda ) ) $, then $J ^m ( \Lambda ) Y = Y J ^m ( \Lambda ) $, so that $ J ^m ( \Lambda ) X = X J ^n ( \Lambda ) $.
\end{proof}

\begin{lemma}
\label{lemma: Lambda centralizer Jm > Jn}
    Let $m$ and $n$ be two positive integers such that $m > n$.
    Let $J ^m ( \Lambda )$ and $J ^n ( \Lambda )$ be Jordan matrices defined in Definition \ref{def: jordan block 1st Lambda}.
    If $X \in \mathbf{M}_{2m, 2n}$ then
    \begin{equation}
        J ^m ( \Lambda ) X = X J ^n ( \Lambda ) \iff X = \left[\begin{array}{c} Y \\ 0 \end{array}\right], Y \in \mathcal{C} ( J ^n ( \Lambda ) )
    \end{equation}
\end{lemma}

\begin{theorem}
\label{theorem: centralizer jordan block 2do Lambda}
    If $J ^{ n _1, \ldots, n _q } ( \Lambda )$ is the Jordan matrix from Definition \ref{def: jordan block 2do Lambda}, then every matrix $ X \in \mathcal{C} ( J ^{ n _1, \ldots, n _q } ( \Lambda ) ) $ is of the form
    \begin{equation}
    \label{matrix in centralizer Lambda 2do}
    X = \left[\begin{array}{ccc}
        X _{1 1}    &   \cdots  &   X _{1 q}    \\
        \vdots      &   \ddots  &   \vdots  \\
        X _{q 1}    &   \cdots  &   X _{q q}
    \end{array}\right]
    \end{equation}
where the following holds:
\begin{itemize}
    \item If $ n _i = n _j $, then $ X _{i j} \in \mathcal{C} ( J ^{n i} ( \Lambda ) ) $,
    \item If $ n _i < n _j $, then $ X _{i j} = \left[\begin{array}{cc} Y & 0 \end{array}\right] $ with $ Y \in \mathcal{C} ( J ^{n _i} ( \Lambda ) ) $,
    \item If $ n _i > n _j $, then $ X _{i j} = \left[\begin{array}{c} Y \\ 0 \end{array}\right] $ with $ Y \in \mathcal{C} ( J ^{n _j} ( \Lambda ) ) $
\end{itemize}
    for all $i , j \in \{ 1, \ldots, q \}$.
\end{theorem}
\begin{proof}
    Let $i , j \in \{ 1, \ldots, q \}$.
    If $X$ is any matrix in $\mathcal{C} ( J ^{ n _1, \ldots, n _q } ( \Lambda ) )$, then $X$ satisfies the equality $J ^{ n _1, \ldots, n _q } ( \Lambda ) X = X J ^{ n _1, \ldots, n _q } ( \Lambda )$.
    Now, we partition the matrix $X$ conformally to $J ^{ n _1, \ldots, n _q } ( \Lambda )$, so $X$ has the form \eqref{matrix in centralizer Lambda 2do} with $X_{i j} \in \mathbf{M}_{2 n_i, 2 n_j}$.
    The intertwining relation of $X$ implies $J ^{n _i} (\Lambda) X_{i j} = X_{i j} J^{n _j} (\Lambda)$.
    If $n_i = n_j$, then $ X _{i j} \in \mathcal{C} ( J ^{n i} ( \Lambda ) ) $.
    Using Lemmas \ref{lemma: Lambda centralizer Jm < Jn} and \ref{lemma: Lambda centralizer Jm > Jn}, we find $ X _{i j} = \left[\begin{array}{cc} Y & 0 \end{array}\right] $ with $ Y \in \mathcal{C} ( J ^{n _i} ( \Lambda ) ) $ for $n_i < n_j$ and $ X _{i j} = \left[\begin{array}{c} Y \\ 0 \end{array}\right] $ with $ Y \in \mathcal{C} ( J ^{n _j} ( \Lambda ) ) $ if $ n _i > n _j $.
\end{proof}

\begin{theorem}
    Let $J \in \mathbf{M}_{m + 2 n} $ be the Jordan matrix from Definition \ref{def: jordan matrix}.
    If $X \in \mathcal{C} ( J )$, then $X$ has the form
    \begin{equation}
    \label{centralizer jordan matrix}
        X = \operatorname{diag} [ 
            Y_1,
            \ldots,
            Y_p,
            Z_1,
            \ldots,
            Z_q
        ] \in \mathbf{M}_{m + 2 n}
    \end{equation}
    where $Y_i \in \mathcal{C} ( J _{m ^i _1, \ldots, m ^i _{r _i}} ( \lambda _i ) )$ and $Z_k \in \mathcal{C} ( J ^{n ^k _1, \ldots, n ^k _{s _k}} ( \Lambda _k ) )$ for all $i \in \{ 1, \ldots, p \}$ and $k \in \{ 1, \ldots, q \}$.
\end{theorem}
\begin{proof}
    Let $i \in \{ 1, \ldots, p \}$ and $k \in \{ 1, \ldots, q \}$.
    If $X \in \mathcal{C} ( J )$, then $X$ verifies $J X = X J$.
    Since all Jordan blocks of the second kind have different eigenvalues, Sylvester's theorem on linear matrix equation \cite{horn_johnson_1985, Gantmacher59} implies that the matrix $X$ has the form given by Eq. \eqref{centralizer jordan matrix}, with $Y_i \in \mathbf{M}_{m_i}$ and $Z_k \in \mathbf{M}_{2 n_k}$.
    Thus,
    \begin{align}
        J _{m ^i _1, \ldots, m ^i _{r _i}} ( \lambda _i ) Y_i & = Y_i J _{m ^i _1, \ldots, m ^i _{r _i}} ( \lambda _i ) \\
        J ^{n ^k _1, \ldots, n ^k _{s _k}} ( \Lambda _k ) Z_k & = Z_k J ^{n ^k _1, \ldots, n ^k _{s _k}} ( \Lambda _k )
    \end{align}
    Theorems \ref{theorem: centralizer jordan block 2do lambda} and \ref{theorem: centralizer jordan block 2do Lambda} imply $Y_i \in \mathcal{C} ( J _{m ^i _1, \ldots, m ^i _{r _i}} ( \lambda _i ) )$ and $Z_k \in \mathcal{C} ( J ^{n ^k _1, \ldots, n ^k _{s _k}} ( \Lambda _k ) )$, respectively.
\end{proof}

This method can be extended to find a set $\{ A_1, \ldots, A_r \}$ of pairwise commuting matrices.
First, choose $A_1$ as the representative of some equivalence class of $\mathfrak{sl} ( n, \mathbb{R} )$.
Then, choose $A_2$ from $\mathcal{C} ( A_1 )$.
Likewise, continue successively until choosing $A_r$ from $\mathcal{C} ( A_1, \ldots, A_{r-1} )$.

\section{Commutative algebra}   %
\label{section: commutative algebras}

The method described in Section \ref{section: centralizer} is very difficult when trying to find a set of three or more pairwise commuting matrices.
Therefore, we propose an alternative technique as follows: $A_1$ continues to be a representative of some equivalence class of $\mathfrak{sl} ( n, \mathbb{R} )$.
If $A_1 \in \mathbf{M}_{n}$ is a non-derogatory matrix, then $\mathcal{C} ( A_1 ) = \mathbb{R} [ A_1 ]$ \cite{47384092-1479-315c-8f62-1c81661484d7}.
This implies that all matrices that commute with $A_1$ are polynomials in $A_1$ over $\mathbb{R}$.
$\mathbb{R} [ A_1 ]$ is a commutative algebra of dimension $n$.
The set of matrices $\{ I_n, A_1, \ldots, A_1^{n-1} \}$ is a basis in $\mathbb{R} [ A_1 ]$ \cite{47384092-1479-315c-8f62-1c81661484d7}.

When $A_1$ is a derogatory matrix, then there exists a non-derogatory matrix $A_2$ that commutes with $A_1$ \cite{Guralnick01061992}.
By Theorem 2 in Chapter 2 of \cite{47384092-1479-315c-8f62-1c81661484d7}, the algebra generated by $A_1$ and $A_2$ over $\mathbb{R}$ has dimension not greater than $n$ and is included in a commutative algebra $\mathfrak{A}$ of dimension $n$.
Since $\{ A_1, A_2 \}$ is a subset of $\mathfrak{A}$, by Lemmas \ref{lemma: subset centralizer} and \ref{lemma: centralizer of pair commuting matrices set}, we obtain $\mathfrak{A} \subset \mathcal{C} ( A_1, A_2 )$.
This means that we can find a commutative algebra of dimension $n$ contained in the centralizer of a representative of some equivalence class of $\mathfrak{sl} ( n, \mathbb{R} )$.
The new method consists of finding it, from which the matrices $A_a$ are selected.
Note that it could happen that the centralizers of two distinct equivalence classes have the same commutative algebra.

The matrices $J_m (\lambda)$ and $J^n (\Lambda)$ due to Definitions \ref{def: jordan block 1st lambda} and \ref{def: jordan block 1st Lambda}, respectively, are non-derogatory.
Their centralizers are given as $\mathfrak{J}_m$ and $\mathfrak{J}^n$, which are commutative algebras of dimension $m$ and $2n$, respectively.

\begin{definition}
\label{def: lambda commutative algebra}
    Let $m_1, \ldots, m_p, m$ be positive integers such that $m_1 + \ldots + m_p = m$.
    Define
    \begin{equation}
        \mathfrak{J} _{m_1, \ldots, m_p} = \{ \operatorname{diag}[ J_1, \ldots, J_p ] \in \mathbf{M}_{m} : J_i \in \mathfrak{J} _{m_i} \text{ for all } i \in \{ 1, \ldots, p \} \}
    \text{.}
    \end{equation}
\end{definition}

Let $m_1, \ldots, m_p, m \in \mathbb{N}$ such that $m = m_1 + \ldots + m_p$ and $\tau_1, \ldots, \tau_p \in \mathbb{R}$ be different constants.
$A_2 = \operatorname{diag} [ J_{ m_1 } ( \tau_1 ), \ldots, J_{ m_p } ( \tau_p ) ] \in \mathbf{M}_m$ is a non-derogatory matrix that commutes with $A_1 = J_{m_1, \ldots, m_p} (\lambda) \in \mathbf{M}_m$, for some $\lambda \in \mathbb{R}$.
If $A_3 \in \mathbf{M}_m$ commutes with $A_1$ and $A_2$, then $A_3 = p ( A_2 ) = \operatorname{diag} [ p ( J_{ m_1 } ( \tau_1 ) ), \ldots, p ( J_{ m_p } ( \tau_p ) ) ]$ for some $p \in \mathbb{R} [ x ]$.
Since $p ( J_{ m_i } ( \tau_i ) )$ is a matrix in $\mathfrak{J}_{ m_i }$ for all $i \in \{ 1, \ldots, p \}$, then $A_1, A_2, A_3 \in \mathfrak{J}_{m_1, \ldots, m_p}$.
It is easy to demonstrate that $\mathfrak{J} _{m_1, \ldots, m_p}$ is a commutative algebra of dimension $m$.
By Lemma \ref{lemma: subset centralizer}, $\mathfrak{J} _{m_1, \ldots, m_p}$ is the largest commutative algebra contained in $\mathcal{C} ( J_{m_1, \ldots, m_p} (\lambda) )$.

\begin{definition}
\label{def: Lambda commutative algebra}
    Let $n_1, \ldots, n_q$ be positive integers such that $n_1 + \ldots + n_q = n$.
    Define
    \begin{equation}
        \mathfrak{J} ^{n_1, \ldots, n_q} = \{ \operatorname{diag}[ J_1, \ldots, J_q ] \in \mathbf{M}_{2 n} : J_i \in \mathfrak{J} ^{n_i} \text{ for all } i \in \{ 1, \ldots, q \} \} 
    \text{.}
    \end{equation}
\end{definition}

Using a similar procedure to find $\mathfrak{J} _{m_1, \ldots, m_p}$, we obtain the largest commutative algebra $\mathfrak{J} ^{n_1, \ldots, n_q}$ of dimension $2 ( n_1 + \ldots + n_q )$ included in $\mathcal{C} ( J^{ n_1, \ldots, n_q } (\Lambda) )$.

\begin{definition}
\label{def: jordan commutative algebra}
    Let $m_1, \ldots, m_p, m, n_1, \ldots, n_q, n$ be positive integers such that $m = m_1 + \ldots + m_p$ and $n = n_1 + \ldots + n_q$.
    Define
    \begin{equation}
        \mathfrak{J} _{m _1, \ldots, m_p} ^{n _1, \ldots, n_q} = \left\{ \operatorname{diag}[ 
            J_1,
            J_2
            ] \in \mathbf{M}_{m + 2 n} :
            J_1 \in \mathfrak{J} _{ m_1, \ldots, m_p},
            J_2 \in \mathfrak{J} ^{n_1, \ldots, n_q}
        \right\}.
    \end{equation}
\end{definition}

The centralizer of any Jordan matrix, as defined in Definition \ref{def: jordan matrix}, contains $\mathfrak{J} _{m _1, \ldots, m_p} ^{n _1, \ldots, n_q}$ as the largest commutative algebra, of dimension is $m _1 + \ldots + m_p + 2 ( n _1 + \ldots + n_q )$.

\begin{lemma}
\label{orthogonality matrices S}
    Let
    \begin{equation}
    \label{def matrices S}
        S_0 = \operatorname{diag}[ I_{n - 1}, -n + 1 ]
    \text{, }
        S_1 = \operatorname{diag}[ I_{n - 2}, -n + 2, 0 ]
    \text{, } \ldots \text{, }
        S_{n - 2} = \operatorname{diag}[ 1, -1, 0_{n - 2} ] \in \mathbf{M}_n
    \text{.}
    \end{equation}
    Then, $\operatorname{tr} S_i S_j = ( n - i - 1 ) ( n - i ) \delta_{i j} $ for all $i, j \in \{ 0, \ldots, n - 2 \}$.
\end{lemma}
\begin{proof}
    Let $i, j \in \{ 0, \ldots, n - 2 \}$ be such that $i < j$.
    So, $S_i^2 = \operatorname{diag}[ I_{n - 1 - i}, ( n - 1 - i )^2, 0_i ] $ implies $\operatorname{tr} S_i^2 = ( n - i - 1 ) ( n - i )$.

    On the other hand, we partition the matrices $S_i$ and $S_j$ as $S_i = \operatorname{diag}[ I_{n - 1 - j}, 1, I_{j - i - 1}, -n + 1 + i, 0_i ] $ and $S_j = \operatorname{diag}[ I_{n - 1 - j}, -n +1 + j, 0_{j - i - 1}, 0, 0_i ]$, respectively.
    Then, $S_i S_j = \operatorname{diag}[ I_{n - 1 - j}, -n +1 + j, 0_{j - i - 1}, 0, 0_i ]$, hence $\operatorname{tr} S_i S_j = 0$.
\end{proof}

Note that $\mathfrak{J}_{ 1, \ldots, 1 }$ is the set of all diagonal matrices in $\mathbf{M}_n$.
A basis for the subspace $\mathfrak{J}_{ 1, \ldots, 1 } \cap \mathfrak{sl} ( n, \mathbb{R} )$ is the set $\{ S_0, \ldots, S_{n - 2} \}$ of pairwise commuting matrices.
By Lemma \ref{orthogonality matrices S}, the set $\{ S_0, \ldots, S_i \}$ is linearly independent for some $i \in \{ 1, \ldots, n - 2 \}$.
We will use this set to build exact solutions in Section \ref{section: example}.

\section{Matrix $g_0$}  
\label{section: subspace ia}

In this section, we will determine the matrix $g_0$.

\begin{definition}
    For any non-empty set $ \mathfrak{A} \subset \mathbf{M}_n$, define
    \begin{equation}
        \mathcal{I} ( \mathfrak{A} ) = \left\{
            X \in \mathbf{Sym}_n :
            A X = X A ^T
            \text{ for all } A \in \mathfrak{A} 
        \right\} \text{.}
    \end{equation}
\end{definition}
\noindent For a singleton set $\{ A \}$, we write $\mathcal{I} ( A )$.
For a finite set $\{ A_1, \ldots, A_r \}$, we write $\mathcal{I} ( A_1, \ldots, A_r )$ or $\mathcal{I} \{ A_a \}$.

\begin{theorem}
    If $\mathfrak A$ is a non-empty subset of $\mathbf{M}_n$, then $\mathcal{I} ( \mathfrak A )$ is a subspace of $\mathbf{M}_n$.
\end{theorem}
\begin{proof}
    Let $X$ and $Y$ be any two matrices in $\mathcal{I} ( \mathfrak A )$ and let $\alpha, \beta \in \mathbb{R}$.
    Then, these matrices verify $A X = X A^T$ and $A Y = Y A^T$, so that $A ( \alpha X + \beta Y ) = ( \alpha X + \beta Y ) A^T$ for all $A \in \mathfrak A$.
    Since $X$ and $Y$ are symmetric matrices, $\alpha X + \beta Y$ is also a symmetric matrix.
    Therefore, $\alpha X + \beta Y \in \mathcal{I} ( \mathfrak A )$.
\end{proof}

\begin{lemma}
\label{lemma: subset ia}
    Let $\mathfrak A, \mathfrak B$ be two non-empty subsets of $\mathbf{M}_{n}$.
    If $\mathfrak{A} \subset \mathfrak{B}$, then $\mathcal{I} ( \mathfrak{B} ) \subset \mathcal{I} ( \mathfrak{A} )$.
\end{lemma}
\begin{proof}
    If $X \in \mathcal{I} ( \mathfrak{B} )$, then $B X = X B^T$ for all $B \in \mathfrak{B}$.
    Since $\mathfrak{A} \subset \mathfrak{B}$, $X$ satisfies $A X = X A^T$ for all $A \in \mathfrak{A}$, so that $X \in \mathcal{I} ( \mathfrak{A} )$.
    Therefore, $\mathcal{I} ( \mathfrak{B} ) \subset \mathcal{I} ( \mathfrak{A} )$.
\end{proof}

\begin{theorem}
\label{theorem: subspace ia union}
    If $\mathfrak A, \mathfrak B$ are non-empty subsets of $\mathbf{M}_{n}$ then
    \begin{equation}
        \mathcal{I} ( \mathfrak{A} \cup \mathfrak{B} ) = \mathcal{I} ( \mathfrak{A} ) \cap \mathcal{I} ( \mathfrak{B} )
    \text{.}
    \end{equation}
\end{theorem}
\begin{proof}
    Since $\mathfrak A$ and $\mathfrak B$ are included in $\mathfrak{A} \cup \mathfrak{B}$, Lemma \ref{lemma: subset ia} implies that $\mathcal{I} ( \mathfrak{A} \cup \mathfrak{B} )$ is contained in $\mathcal{I} ( \mathfrak{A} )$ and $\mathcal{I} ( \mathfrak{B} )$.
    Thus, $\mathcal{I} ( \mathfrak{A} \cup \mathfrak{B} ) \subset \mathcal{I} ( \mathfrak{A} ) \cap \mathcal{I} ( \mathfrak{B} )$.
    Now, if $X \in \mathcal{I} ( \mathfrak{A} ) \cap \mathcal{I} ( \mathfrak{B} )$, then $X \in \mathcal{I} ( \mathfrak{A} )$ and $X \in \mathcal{I} ( \mathfrak{B} )$, so that $A X = X A^T$ for all $A \in \mathfrak{A}$ and $B X = X B^T$ for all $B \in \mathfrak{B}$.
    The intertwining relations can be written as $C X = X C^T$ for all $C \in \mathfrak{A} \cup \mathfrak{B}$, then $X \in \mathcal{I} ( \mathfrak{A} \cup \mathfrak{B} )$. 
    Consequently, $\mathcal{I} ( \mathfrak{A} ) \cap \mathcal{I} ( \mathfrak{B} ) \subset \mathcal{I} ( \mathfrak{A} \cup \mathfrak{B} )$.
\end{proof}

\begin{lemma}
\label{lemma: ia non-derogatory matrix}
    Each matrix $A$ in $\mathbf{M}_{n}$ satisfies $\mathcal{I} ( A ) = \mathcal{I} ( \mathbb{R} [A] )$.
\end{lemma}
\begin{proof}
    If $X \in \mathcal{I} ( A )$, then $A X = X A^T$, so that $p( A ) X = X ( p ( A ) )^T$ for any polynomial $p$ in $A$ over $\mathbb{R}$.
    Therefore, $X \in \mathcal{I} ( \mathbb{R} [A] )$, which implies $\mathcal{I} ( A ) \subset \mathcal{I} ( \mathbb{R} [A] )$.
    On the other hand, since $A\in \mathbb{R} [A]$, Lemma \ref{lemma: subset ia} provides $\mathcal{I} ( \mathbb{R} [A] ) \subset \mathcal{I} ( A )$.
\end{proof}

According to Lemma \ref{lemma: solution diff eq g}, $g_0$ satisfies the intertwining relations $A_a g_0 = g_0 A_a^T$ for all $a$.
Then, $g_0$ belongs to the sets $\mathcal{I} ( A_a )$, so that $g_0$ is a matrix in $\mathcal{I} \{ A_a \}$ as a consequence of Theorem \ref{theorem: subspace ia union}.

When $A_1$ is a non-derogatory matrix, the commutative algebra is $\mathfrak{A} = \mathbb{R}[A_1]$.
By Lemma \ref{lemma: ia non-derogatory matrix}, we obtain $\mathcal{I} (\mathfrak{A} ) = \mathcal{I} ( A_1 )$.
Hence, $\mathcal{I} ( \mathfrak{J} _m ) = \mathcal{I} ( J_m(\lambda) )$ and $\mathcal{I} ( \mathfrak{J} ^n ) = \mathcal{I} ( J ^n(\Lambda) )$ for the eigenvalues $\lambda$ and $\Lambda$.
In \cite{Sarmiento-Alvarado2023}, the authors compute the sets $\mathcal{I} ( A )$ for the five types of equivalence classes of $\mathfrak{sl} ( n, \mathbb{R} )$.


\begin{theorem}
\label{theo: ia jordan block 1st lambda}
    Let $J _m ( \lambda )$ be a Jordan block due to Definition \ref{def: jordan block 1st lambda}. Then
    \begin{equation}
        \mathcal{I} ( J _m ( \lambda ) ) =\left\{ A\in \mathbf{M}_{m} : A= \left[\begin{array}{cccc}
                x_1 & \cdots & x_{m-1} & x_m \\
                x_2 & \cdots & x_m    &\\
                \vdots & \iddots &&\\
                x _m &&&
            \end{array}\right] 
        \right\}
    \end{equation}
\end{theorem}

\begin{theorem}
\label{theo: ia jordan block 2nd lambda}
    Let $J _{ m_1, \ldots, m_p } ( \lambda )$ be a Jordan block due to Definition \ref{def: jordan block 2do lambda}.
    Every matrix $X \in \mathcal{I}( J _{ m_1, \ldots, m_p } ( \lambda ) )$ is a block matrix of the form
    \begin{equation}
        X = \left[\begin{array}{ccc}
            X _{1 1}    &   \cdots  &   X _{1 p}    \\
            \vdots      &   \ddots  &   \vdots  \\
            X _{p 1}    &   \cdots  &   X _{p p}
        \end{array}\right]
    \end{equation}
    where for each $i, j \in \{ 1, \ldots, p \}$, $ X _{i j} \in \mathbf{M}_{m _i, m _j} $ satisfies $ X ^T _{i j}  = X _{j i}  $ and is of the following form:
    \begin{enumerate}
        \item   If $m _i = m _j$, then $X _{ij} \in \mathcal{I}( J _{m _i} (\lambda) )$.
        \item   If $m _i < m _j$, then $ X _{ij} = \left[\begin{array}{cc} Y _{ij} & 0 \end{array}\right] $ with $Y _{ij} \in \mathcal{I}( J _{m _i} (\lambda) )$.
        \item   If $m _i > m _j$, then $ X _{ij} = \left[\begin{array}{c} Y _{ij} \\ 0 \end{array}\right] $ with $Y _{ij} \in \mathcal{I}( J _{m _j} (\lambda) )$.
    \end{enumerate}
\end{theorem}


\begin{theorem}
\label{theo: ia jordan block 1st Lambda}
    Let $J ^n ( \Lambda )$ be a Jordan block due to Definition \ref{def: jordan block 1st Lambda}. Then
    \begin{equation}
        \mathcal{I} ( J ^n ( \Lambda ) ) = \left\{
            \left[\begin{array}{cccc}
                X_1 & \cdots & X_{m-1} & X_m \\
                X_2 & \cdots & X_m    &\\
                \vdots & \iddots &&\\
                X_m &&&
            \end{array}\right] \in \mathbf{M}_{2 n}
            :\  X _1, \ldots, X _n \in \mathcal{I} ( \Lambda )
        \right\},
    \end{equation}
    where
    \begin{equation}
        \mathcal{I} ( \Lambda ) = \left[\begin{array}{cc}
            a & b  \\
            b & -a 
        \end{array}\right] \in \mathbf{M}_2.
    \end{equation}
\end{theorem}

\begin{theorem}
\label{theo: ia jordan block 2nd Lambda}
    Let $J ^{ n_1, \ldots, n_q } ( \Lambda )$ be a Jordan block due to Definition \ref{def: jordan block 2do Lambda}.
    Then every matrix $X \in \mathcal{I}( J ^{ n _1, \ldots, n _q } ( \Lambda ) )$ is a block matrix of the form
    \begin{equation}
        X = \left[ \begin{array}{ccc}
            X _{1 1}    &   \cdots  &   X _{1 q}    \\
            \vdots      &   \ddots  &   \ddots  \\
            X _{q 1} &   \cdots  &   X _{q q}
        \end{array}\right]
    \end{equation}
    where for each $i, j \in\{ 1, \ldots, q \} $, $X _{i j} \in \mathbf{M}_{2n_i, 2n_j}$ and $X _{j i} = X ^T _{i j} $ which are of the following form:
    \begin{enumerate}
        \item   If $n_i = n_j$, then $X_{ij} \in \mathcal{I}( J ^{n _i} ( \Lambda ) )$.
        \item   If $n_i < n_j$, then $X_{ij} = \left[ \begin{array}{cc} Y_{ij} & 0 \end{array}\right] $ with $Y_{ij} \in\mathcal{I}(J ^{n_i} (\Lambda ) )$.
        \item   If $n_i > n_j$, then $X_{ij} = \left[ \begin{array}{c} Y_{ij} \\ 0 \end{array}\right] $ with $Y_{ij} \in\mathcal{I}(J ^{n_j} (\Lambda ) )$.
    \end{enumerate}
\end{theorem}


\begin{theorem}
\label{theo: ia jordan matrix}
    Let $ J $ be a Jordan matrix due to Definition \ref{def: jordan matrix}.
    If $Z \in \mathbf{M}_{m+2n}$ is a matrix in $\mathcal{I} ( J )$, then $Z$ has the form $\operatorname{diag} ( X_1, \ldots, X_p, Y _1, \ldots, Y_q )$, where $X_i \in \mathcal{I} ( J _{m ^i _1, \ldots, m ^i _{r _i}}(\lambda_i ) )$ for all $i \in \{ 1, \ldots, p \}$ and $Y_k \in \mathcal{I} ( J ^{n ^j _1, \ldots, n ^k _{s _k}}(\Lambda_k ) )$ for all $k \in \{ 1, \ldots, q \}$.
\end{theorem}

If $A_1$ is a derogatory matrix, then $\{ A_a \} \subset \mathfrak{A}$.
Thus, by Lemma \ref{lemma: subset ia}, $\mathcal{I} ( \mathfrak{A} ) \subset \mathcal{I} \{ A_a \}$.
This means that any matrix in $\mathcal{I} ( \mathfrak{A} )$ is also a matrix in $\mathcal{I} \{ A_a \}$.
For this reason, we will determine the subspaces $\mathcal{I} (  \mathfrak{J} _{m_1, \ldots, m_p} )$, $\mathcal{I} (  \mathfrak{J} ^{n_1, \ldots, n_q} )$, and $\mathcal{I} ( \mathfrak{J} _{m_1, \ldots, m_p} ^{n_1, \ldots, n_q} )$.

\begin{theorem}
\label{theorem: ia jordan lambda-algebra}
    Let $\mathfrak{J} _{m_1, \ldots, m_p}$ be the commutative algebra given in Definition \ref{def:  lambda commutative algebra}.
    Then,
    \begin{equation}
        \mathcal{I} ( \mathfrak{J} _{m_1, \ldots, m_p} ) = \left\{
            \operatorname{diag}[ g_1, \ldots, g_p ] \in \mathbf{M}_{m} :
            g_i \in \mathcal{I} ( \mathfrak{J} _{m_i} ) \text{ for all } i \in \{ 1, \ldots, p \}
        \right\}
    \text{.}
    \end{equation}
\end{theorem}
\begin{proof}
    Let $i \in \{ 1, \ldots, p \}$.
    If the matrix $g$ belongs to $\mathcal{I} ( \mathfrak{J} _{m_1, \ldots, m_p} )$, it can be partitioned into matrices of $\mathfrak{J} _{m_1, \ldots, m_p}$:
    \begin{equation}
        g = \left[\begin{array}{ccc}
            g_{1 1} & \cdots & g_{1 p} \\
            \vdots & \ddots & \vdots \\
            g_{p 1} & \cdots & g_{p p} \\
        \end{array}\right]
    \text{,}
    \end{equation}
    where $g_{i j} \in \mathbf{M}_{m_i, m_j}$.
    The intertwining relation $A g = g A^T$ for all $A \in \mathfrak{J} _{m_1, \ldots, m_p}$, implies that $A_i g_{i j} = g_{i j} A_j ^T$ for all $A_i \in \mathfrak{J} _{m_i}$.
    Since the matrices $A_i$ are arbitrary, they can have different real eigenvalues. By Sylvester’s theorem on linear matrix equations \cite{horn_johnson_1985,Gantmacher59}, we get $g_{i j} = 0$ for $i \neq j$.
    Furthermore, by the symmetry of $g$, $g ^T = g$, the matrices $g_{i i}$ are symmetric, so that $g_{i i} \in \mathcal{I} ( \mathfrak{J} _{m_i} )$.
\end{proof}

\begin{theorem}
\label{theorem: ia jordan Lambda-algebra}
    If $\mathfrak{J} ^{n_1, \ldots, n_q}$ is the commutative algebra defined in Definition \ref{def: Lambda commutative algebra}, then
    \begin{equation}
        \mathcal{I} ( \mathfrak{J} ^{n_1, \ldots, n_q} ) = \left\{
            \operatorname{diag}[ g_1, \ldots, g_q ] \in \mathbf{M}_{m} :
            g_i \in \mathcal{I} ( \mathfrak{J} ^{n_i} ) \text{ for all } i \in \{ 1, \ldots, q \}
        \right\}.
    \end{equation}
\end{theorem}

\begin{theorem}
\label{theorem: ia jordan algebra}
    Let $\mathfrak{J} _{m_1, \ldots, m_p} ^{n_1, \ldots, n_q}$ be the commutative algebra introduced in Definition \ref{def: jordan commutative algebra}.
    Then,
    \begin{equation}
        \mathcal{I} ( \mathfrak{J} _{m_1, \ldots, m_p} ^{ n_1, \ldots, n_q } ) = \left\{
            \operatorname{diag}[ g_\lambda, g_\Lambda ] \in \mathbf{M}_{m + 2 n} :
            g_\lambda \in \mathcal{I} ( \mathfrak{J} _{ m_1, \ldots, m_p } ),
            g_\Lambda \in \mathcal{I} ( \mathfrak{J} ^{ n_1, \ldots, n_q } )
        \right\}.
    \end{equation}
\end{theorem}
\begin{proof}
    If the matrix $g$ belongs to $\mathcal{I} ( \mathfrak{J} _{m_1, \ldots, m_p} ^{n_1, \ldots, n_q} )$, it can be partitioned into matrices from $\mathfrak{J} _{m_1, \ldots, m_p} ^{n_1, \ldots, n_q}$, as follows:
    \begin{equation}
        g = \left[\begin{array}{cc}
            g_{1 1} & g_{1 2} \\
            g_{1 2}^T & g_{2 2}
        \end{array}\right],
    \end{equation}
    where $g_{1 1} \in \mathbf{Sym}_{m}$, $g_{2 2} \in \mathbf{Sym}_{2 n}$ and $g_{1 2} \in \mathbf{M}_{m, 2 n}$.
    The intertwining relation $A g = g A^T$ for all $A \in \mathfrak{J} _{m_1, \ldots, m_p} ^{n_1, \ldots, n_q}$, implies that $P g_{1 1} = g_{1 1} P^T$, $Q g_{2 2} = g_{2 2} Q^T$ and $P g_{1 2} = g_{1 2} Q^T$ for all $P \in \mathfrak{J} _{ m_1, \ldots, m_p }$ and $Q \in \mathfrak{J} ^{ n_1, \ldots, n_q }$.
    Thus, $g_{1 1} \in \mathcal{I} ( \mathfrak{J} _{ m_1, \ldots, m_p } )$ and $g_{2 2} \in \mathcal{I} ( \mathfrak{J} ^{ n_1, \ldots, n_q } )$.
    Since the matrices in $\mathfrak{J} _{ m_1, \ldots, m_p }$ have real eigenvalues, and the matrices in $\mathfrak{J} ^{ n_1, \ldots, n_q }$ have complex conjugate eigenvalues, by Sylvester’s theorem on linear matrix equations \cite{horn_johnson_1985,Gantmacher59}, we obtain that $g_{1 2} = 0$.
\end{proof}

\begin{lemma}
\label{lemma: subspace ia matrices S}
    Let $S_0, \ldots, S_{n - 2} \in \mathbf{M}_n$ be the matrices defined in Eq. \eqref{def matrices S}.
    If $g$ is a matrix in $\mathcal{I} ( S_0, \ldots, S_i )$ for some $i \in \{ 1, \ldots, n - 2 \}$, then $g$ has the form
    \begin{equation}
    \label{subspace ia matrices S}
        g = \operatorname{diag}[ M, c_1, \ldots, c_{i + 1} ],
    \end{equation}
    where $M \in \mathbf{Sym}_{ n - 1 - i }$ and $c_1, \ldots, c_{i + 1} \in \mathbb{R}$. 
\end{lemma}
\begin{proof}
    If $g \in \mathcal{I} ( S_0 )$, then $S_0 g = g S_0$, which implies $g = \operatorname{diag}[ R, C ]$, where $R \in \mathbf{Sym}_{n - 1}$ and $C \in \mathbb{R}$.

    Now, let us assume that Eq. \eqref{subspace ia matrices S} is also true for $i < n - 2$.
    We prove that it is also true for $i + 1$.
    If $g \in \mathcal{I} ( S_0, \ldots, S_{i + 1} )$, then $g$ is also a matrix in $g \in \mathcal{I} ( S_0, \ldots, S_i )$, so that $g$ has the form given by Eq. \eqref{subspace ia matrices S}.
    Moreover, the matrix $g$ belongs to $\mathcal{I} ( S_{i + 1} )$, then $\operatorname{diag}[ I_{ n - i - 2 }, -n + i + 2, 0_{ i + 1} ] \operatorname{diag}[ M, c_1, \ldots, c_{i + 1} ] = \operatorname{diag}[ M, c_1, \ldots, c_{i + 1} ] \operatorname{diag}[ I_{ n - i - 2 }, -n + i + 2, 0_{ i + 1} ]$, so that $\operatorname{diag}[ I_{ n - i - 2 }, -n + i + 2 ] M = M \operatorname{diag}[ I_{ n - i - 2 }, -n + i + 2 ]$.
    By Sylvester’s theorem on linear matrix equations \cite{horn_johnson_1985,Gantmacher59}, we have $M = \operatorname{diag}[ N, c ]$, where $N \in \mathbf{Sym}_{ n - i - 2 }$ and $c \in \mathbb{R}$.
    Therefore, $g = \operatorname{diag}[ N, c, c_1, \ldots, c_{i + 1} ]$.
\end{proof}

\section{Exponential matrix}    
\label{section: matrix exponential}

Determining the exponential of a matrix is a long and complex process.
However, it is necessary to construct exact solutions to the EFE, since $g$ is given by the exponential of a linear combination of matrices $A_a$.
In this section, we will compute the exponential of polynomials of Jordan matrices.

\begin{theorem}
\label{theorem: upper triangular syst lin matrix diff eq}
    Let $\{ M_1, \ldots, M_n \}$ be a subset of pairwise commuting matrices of $\mathbf{M}_m$.
    Then,
    \begin{equation}
        \exp \left[\begin{array}{cccc}
            M_1 & M_2 & \cdots & M_n \\
            & M_1 & \cdots & M_{n - 1} \\
            && \ddots & \vdots \\
            &&& M_1 \\
        \end{array}\right]
        = \left[\begin{array}{cccc}
            E_1 & E_2 & \cdots & E_n \\
            & E_1 & \cdots & E_{n - 1} \\
            && \ddots & \vdots \\
            &&& E_1 \\
        \end{array}\right] \text{,} 
    \end{equation}
    where
    \begin{align}
        E_1
    &   = e^{M_1} \text{, }
        E_i
        = e^{M_1} \sum_{j = 1}^{i - 1} \frac{C_{i, j}}{j !} \text{,}
    \\\label{recurrence relation C}
        C_{1, 0}
    &   = I_m \text{, }
        C_{i, 0} = 0_m \text{, }
        C_{i, l + 1}
        = \sum_{j = 2}^{i - l} M_j C_{i - j + 1, l}
    \end{align}
    for all $i \in \{ 2, \ldots, n \}$ and $l \in \{0, \ldots, i - 2\}$.
\end{theorem}
\begin{proof}
    Let
    \begin{equation}
        M = \left[\begin{array}{cccc}
            M_1 & M_2 & \cdots & M_n \\
            & M_1 & \cdots & M_{n - 1} \\
            && \ddots & \vdots \\
            &&& M_1 \\
        \end{array}\right] \in \mathbf{M}_{n m} \text{.}
    \end{equation}
    Let $E (\sigma ) \in \mathbf{M}_{nm}$ be a matrix function that satisfies the linear differential equation $E_{, \sigma } = M E$ with the initial condition $E ( 0 ) = I_{n m}$.
    Then, $E (\sigma )$ has the same form as $M$, that is,
    \begin{equation}
        E = \left[\begin{array}{cccc}
            E_1 & E_2 & \cdots & E_n \\
            & E_1 & \cdots & E_{n - 1} \\
            && \ddots & \vdots \\
            &&& E_1 \\
        \end{array}\right] \text{,}
    \end{equation}
    where $E_1 (\sigma), \ldots, E_n (\sigma) \in \mathbf{M}_m$.
    From the linear differential equation of $E$, we have $E_{i , \sigma} = \sum_{j = 1}^i M_j E_{i - j + 1}$ for all $i \in \{ 1, \ldots, n \}$.
    Let $X_i ( \sigma ) \in \mathbf{M}_m$ be matrix functions.
    We propose $E_i ( \sigma ) = e^{\sigma M_1} X_i ( \sigma )$ to obtain
    \begin{align}
    \label{diff eq X1}
        \frac{d X_1}{d \sigma}
    &   = 0 \text{,}
    \\\label{diff eq Xp}
        \frac{d X_i}{d \sigma}
    &   = \sum_{l = 2}^i M_l X_{i - l + 1} \text{ for } i > 1 \text{.}
    \end{align}
    Now, we will demonstrate that
    \begin{equation}
    \label{pth der Xp}
        \frac{d^i X_i}{d \sigma^i} = 0
    \end{equation}
    by mathematical induction.
    For $i = 1$, it is true by Eq. \eqref{diff eq X1}.
    Let us assume that Eq. \eqref{pth der Xp} is also true for $1 < i < n$.
    We prove that it holds for $i + 1$.
    Differentiating Eq. \eqref{diff eq Xp} $i$ times and using the hypothesis \eqref{pth der Xp}, we obtain $\frac{d^{i + 1} X_{i + 1}}{d \sigma^{i + 1}} = \sum_{l = 2}^{i + 1} M_l \frac{d^i X_{i - l + 2}}{d \sigma^i } = 0$.

    Solving Eq. \eqref{pth der Xp}, we find
    \begin{equation}
    \label{Xi functions}
        X_i (\sigma) = \sum_{j = 0}^{i - 1} C_{i, j} \frac{\sigma^j}{j!} \text{,}
    \end{equation}
    where $C_{i, j}$ are constant matrices in $\mathbf{M}_m$ for all $j \in \{ 0, \ldots, i - 1 \}$.
    To determine the constant matrices of $X_i$ for $i > 1$, we substitute Eq. \eqref{Xi functions} into Eq. \eqref{diff eq Xp}.
    So,
    \begin{equation}
        \sum_{l = 0}^{i - 2} C_{i, l + 1} \frac{\sigma^l}{l!}
        = \sum_{j = 2}^i \sum_{l = 0}^{i - j} M_j C_{i - j + 1, l} \frac{\sigma^l}{l!}
        = \sum_{l = 0}^{i - 2} \sum_{j = 2}^{i - l} M_j C_{i - j + 1, l} \frac{\sigma^l}{l!}
        \text{ for } i > 1
    \end{equation}
    From it, we get the recurrence relation $\eqref{recurrence relation C}$.
    The initial condition implies that $E_1 ( 0 ) = I_m$ and $E_i ( 0 ) = 0_m$ for $i > 1$.
    Then, $X_1 ( 0 ) = I_m$ and $X_i ( 0 ) = 0_m$, so that $C_{1,0} = I_m$ and $C_{i, 0} = 0_m$ for $i > 1$.
    Therefore, $E_1 ( \sigma ) = e^{\sigma M_1 }$ and $E_i ( \sigma ) = e^{\sigma M_1 } \sum_{j = 1}^{i - 1} C_{i, j} \frac{\sigma^j}{j!}$ for $i > 1$.
    Since  $E (\sigma ) = e^{\sigma M}$, we only need to evaluate the matrix functions $E_i$ at $\sigma = 1$ to obtain the blocks of $e^{M}$.
\end{proof}

Applying $C_{1,0} = I_m$ and $C_{i, 0} = 0_m$ to Eq. \eqref{recurrence relation C} with $l = 0$, we obtain $C_{i, 1} = M_i$ for $i > 1$.
Substituting this result into Eq. \eqref{recurrence relation C} with $l = 1$, we determine $C_{i, 2} = \sum_{j = 2}^{i - 1} M_j M_{i - j + 1}$ for $i > 2$, which in turn allows us to calculate $E_2 = e^{M_1} M_2$ and $E_3 = e^{M_1} ( M_3 + 1/2 M_2^2 )$.

\begin{lemma}
\label{lemma: g lambda-alg 1st}
    Let $\lambda_1, \ldots, \lambda_m \in \mathbb{R}$.
    Then,
    \begin{equation}
        \exp \left[\begin{array}{cccc}
            \lambda_1 & \lambda_2 & \cdots & \lambda_m \\
            & \lambda_1 & \cdots & \lambda_{m - 1} \\
            && \ddots & \vdots \\
            &&& \lambda_1 \\
        \end{array}\right]
        = \left[\begin{array}{cccc}
            E_1 & E_2 & \cdots & E_m \\
            & E_1 & \cdots & E_{m - 1} \\
            && \ddots & \vdots \\
            &&& E_1 \\
        \end{array}\right] \text{,} 
    \end{equation}
    where
    \begin{equation}        
        E_1 = e^{\lambda_1} \text{,}
        E_i = e^{\lambda_1} \sum_{j = 1}^{i - 1} \frac{C_{i, j}}{j !} \text{,}
        C_{1, 0} = 1 \text{,}
        C_{i, 0} = 0 \text{,}
        C_{i, l + 1}
        = \sum_{j = 2}^{i - l} \lambda_j C_{i - j + 1, l}
    \end{equation}
    for all $i \in \{ 2, \ldots, m \}$, $l \in \{0, \ldots, i - 2\}$.
\end{lemma}

\begin{lemma}
\label{lemma: g Lambda-alg 1st}
    Let
    \begin{equation}
        \Lambda_1 = \left[\begin{array}{rr}
             a_1 & -b_1  \\
             b_1 & a_1 \\ 
        \end{array}\right], \ldots,
        \Lambda_n = \left[\begin{array}{rr}
             a_n & -b_n  \\
             b_n & a_n \\ 
        \end{array}\right] \in \mathbf{M}_2 \text{.}
    \end{equation}
    Then,
    \begin{equation}
        \exp \left[\begin{array}{cccc}
            \Lambda_1 & \Lambda_2 & \cdots & \Lambda_n \\
            & \Lambda_1 & \cdots & \Lambda_{n-1} \\
            && \ddots & \vdots \\
            &&& \Lambda_1 \\
        \end{array}\right]
        = \left[\begin{array}{cccc}
            G_1 & G_2 & \cdots & G_n \\
            & G_1 & \cdots & G_{n-1} \\
            && \ddots & \vdots \\
            &&& G_1 \\
        \end{array}\right] ,
    \end{equation}
    where
    \begin{equation}
    \label{U and V}
    \begin{aligned}
        G_1
    &   =\left[\begin{array}{rr}
            e^{a_1} \cos b_1 & -e^{a_1} \sin b_1  \\
            e^{a_1} \sin b_1 & e^{a_1} \cos b_1 \\ 
        \end{array}\right],
    &   G_i
    &   =\left[\begin{array}{rr}
            E_i & -F_1  \\
            F_1 & E_1 \\ 
        \end{array}\right],
    \\  E_i & = \sum _{j = 1} ^{i - 1} \frac{e^{a_1}}{j !} \left( C_{i, j} \cos b_1 - D_{i, j} \sin b_1 \right),
    &   F_i & = \sum _{j = 1} ^{i - 1} \frac{e^{a_1} }{j !} \left( D_{i, j} \cos b_1 + C_{i, j} \sin b_1 \right),
    \\  C_{1,0} & = 1 ,
    &   D_{1,0} & = 0 ,
    \\  C_{i,0} & = 0 ,
    &   D_{i,0} & = 0 ,
    \\  C _{i, l + 1} & = \sum _{j = 2} ^{i - l} \left( a_j C_{i - j + 1, l} - b_j D_{i - j + 1, l} \right) ,
    &   D _{i, l + 1} & = \sum _{j = 2} ^{i - l} \left( b_j C_{i - j + 1, l} + a_{j + 1} D_{i - j + 1, l} \right)
    \end{aligned}
    \end{equation}
    for all $i \in \{ 2, \ldots, n \}$, $l \in \{ 0, \ldots, n - 2 \}$.
\end{lemma}

\begin{lemma}
\label{lemma: exp matrices S}
    Let $S_0, \dots, S_{n - 2}$ be the matrices defined in Eq. \eqref{def matrices S} and let $\lambda_0, \ldots, \lambda_{n - 2} \in \mathbb{R}$.
    Then,
    \begin{equation}
    \label{exp matrices S}
        e^{ \lambda_0 S_0 + \ldots + \lambda_i S_i } = \operatorname{diag} [ e^{ \lambda_0 + \ldots + \lambda_i } I_{ n - 1 - i }, e^{ \lambda_0 + \ldots + \lambda_{ i - 1} - (n - 1 - i) \lambda_i }, \ldots, e^{ -(n - 1) \lambda_0 } ]
    \end{equation}
    for some $i \in \{ 1, \ldots, n - 2 \}$.
\end{lemma}
\begin{proof}
    Let $E_0 = e^{ \lambda_0 S_0}$ and $E_i = e^{ \lambda_0 S_0 + \ldots + \lambda_i S_i }$ for all $i \in \{ 1, \ldots, n - 2 \}$.
    We will demonstrate Eq. \eqref{exp matrices S} by mathematical induction.
    The exponential of $S_0$ is $e^{ \lambda_0 S_0 } = \operatorname{diag} [ e^{ \lambda_0 } I_{ n - 1 }, e^{ - (n - 1) \lambda_0 } ]$.
    Assume that Eq. \eqref{exp matrices S} holds for some value $i < n - 2$, we will prove that it holds for $i + 1$.
    We partition $e^{ \lambda_1 S_1 + \ldots + \lambda_i S_i }$ as
    \begin{equation}
        E_i = \operatorname{diag} [ e^{ \lambda_0 + \ldots + \lambda_i } I_{ n - 2 - i }, e^{ \lambda_0 + \ldots + \lambda_i }, e^{ \lambda_0 + \ldots + \lambda_{ i - 1} - (n - 1 - i) \lambda_i }, \ldots, e^{ -(n - 1) \lambda_0 } ]
    \text{.}
    \end{equation}
    Then,
    \begin{equation}        
    \begin{aligned}
        E_{ i + 1}
    &   = E_i e^{ \lambda_{ i + 1 } S_{ i + 1 } }
        = E_i \operatorname{diag} [ e^{ \lambda_{ i + 1 } } I_{ n - 2 - i }, e^{ - (n - 2 - i) \lambda_{ i + 1 } }, I_{ i + 1 } ]
    \\& = \operatorname{diag} [ e^{ \lambda_0 + \ldots + \lambda_{ i + 1} } I_{ n - 2 - i }, e^{ \lambda_0 S_0 + \ldots + \lambda_i - (n - 2 - i) \lambda_{ i + 1} }, e^{ \lambda_0 + \ldots + \lambda_{ i - 1} - (n - 1 - i) \lambda_i }, \ldots, e^{ -(n - 1) \lambda_0 }  ]
    \text{.}
    \end{aligned}
    \end{equation}
\end{proof}

\section{Examples}  %
\label{section: example}

In this section, we will utilize the mathematical tools developed to build exact solutions to the EFE.

In order to build an exact solution to the EFE using our method, we need a set $\{ A_a \}$ of pairwise commuting matrices, a constant matrix $g_0$ that belongs to $\mathcal{I} \{ A_a \}$, a set $\{ \xi^a \}$ of solutions to the generalized Laplace equation \eqref{gen Laplace eq} and a function $f$, which is a solution to Eq. \eqref{SL invariant field eq f}.
To determine the set $\{ A_a \}$, we introduce two methods: the first is ideal for two matrices, and the second consists of finding a commutative algebra contained in an equivalence class of $\mathfrak{sl} ( n, \mathbb{R} )$.
Using the method introduced by \cite{Sarmiento-Alvarado2023}, we compute the equivalence classes of $\mathfrak{sl} ( 3, \mathbb{R} )$ and $\mathfrak{sl} ( 4, \mathbb{R} )$ and report them in Tables \ref{table: eq classes sl3} and \ref{table: eq classes sl4}, respectively.
The matrices $A_1$, $A_2$ and $g_0$ are shown in second, third, and fourth columns of both tables.
The commutative algebras $\mathfrak{A}$ contained in $\mathfrak{sl} ( 3, \mathbb{R} )$ and $\mathfrak{sl} ( 4, \mathbb{R} )$ are shown in Tables \ref{table: comm alg sl3} and \ref{table: comm alg sl4}, respectively.
Specifically, there are four such algebras in $\mathfrak{sl} ( 3, \mathbb{R} )$ and nine in $\mathfrak{sl} ( 4, \mathbb{R} )$.
The second column of these tables shows the exponentials of the representatives of these commutative algebras.

In what follows, we will determine the parameters $\xi^a$ and the function $f$.
To solve the generalized Laplace equation, we write it in terms of the Boyer-Lindquist coordinates:
\begin{equation}
\label{def Boyer-Lindquist coordinates}
    \rho = \sqrt{ r^2 - 2 m r + \sigma^2 } \sin \theta
    \text{ and }
    \zeta = ( r - m ) \cos \theta
\text{,}
\end{equation}
where $m$ and $\sigma$ are constant parameters, $0 < r$ and $0 < \theta < \pi$.
Then, it takes the form 
\begin{equation}
\label{Laplace eq Boyer-Lindquist coordinates}
    ( ( r^2 - 2 m r + \sigma^2 ) \xi^a_{, r} ) _{, r} + \frac{1}{\sin \theta} ( \xi^a_{, \theta} \sin \theta ) _{, \theta} = 0
\text{.}
\end{equation}
If $\xi^a ( r, \theta ) = X^a (r)+ Y^a ( \theta )$, then the solutions of Eq. \eqref{Laplace eq Boyer-Lindquist coordinates} are
\begin{align}
\label{par sol m2 geq}
    \xi^a ( r, \theta )
&   = \mathscr{C}^a \ln \rho 
    + \frac{\mathscr{D}^a}{2 \sqrt{ m^2 - \sigma^2 }} \ln \frac{r_-}{r_+}
    + \mathscr{E}^a \ln \tan \frac{\theta}{2}
    + \xi^a_0
&   \text{for } m^2 > \sigma^2,
\\\label{par sol m2 leq}
    \xi^a ( r, \theta )
&   = \mathscr{C}^a \ln \rho 
    + \frac{\mathscr{D}^a}{\sqrt{ \sigma^2 - m^2 }} \Psi (r)
    + \mathscr{E}^a \ln \tan \frac{\theta}{2}
    + \xi^a_0
&   \text{for } m^2 < \sigma^2,
\\\label{par sol m2 eq}
    \xi^a ( r, \theta )
&   = \mathscr{C}^a \ln \rho 
    - \frac{\mathscr{D}^a}{r - m}
    + \mathscr{E}^a \ln \tan \frac{\theta}{2}
    + \xi^a_0
&   \text{for } m^2 = \sigma^2,
\end{align}
where $\mathscr{C}^a, \mathscr{D}^a, \mathscr{E}^a, \xi^a_0$ are real constants, $\Psi (r) = \tan^{-1} \frac{r - m}{\sqrt{\sigma^2 - m^2}}$ and $r_\pm = r - m \pm \sqrt{m^2 - \sigma^2}$.
Once we know $\xi^a ( r, \theta )$, we determine $f$ by solving Eqs. \eqref{SL invariant field eq f}, which, in Boyer-Lindquist coordinates, change to
\begin{equation}\label{diff eq f Boyer-Lindquist coordinates}
\begin{aligned}
    \left( \ln f \rho^{1-1/n} \right)_{, r} & = \frac{\operatorname{tr} A_a A_b}{4} \frac{
        (r^2 -2 m r + \sigma^2) (r - m) \sin ^2 \theta
    }{
        ( r - m )^2 + ( \sigma^2 - m^2 ) \cos^2 \theta
    } \left[
         \xi^a_{, r} \xi^b_{, r} 
    + \frac{
            2 \cot \theta \xi^a_{, r} \xi^b_{, \theta}
        }{
            r - m
        }
    \right.
\\& \left.
        - \frac{
            \xi^a_{, \theta} \xi^b _{, \theta}
        }{
            r^2 -2 m r + \sigma^2
        }
    \right]
\text{,}
\\  \left( \ln f \rho^{1-1/n} \right)_{, \theta}
&   = - \frac{\operatorname{tr} A_a A_b}{4} \frac{
        (r^2 -2 m r + \sigma^2)^2 \sin \theta \cos \theta
    }{
        ( r - m )^2 + ( \sigma^2 - m^2 ) \cos^2 \theta
    } \left[
        \xi^a_{, r} \xi^b_{, r}
    - \frac{
            2 (r - m) \tan \theta \xi^a_{, r} \xi^b_{, \theta}
        }{
            r^2 -2 m r + \sigma^2
        }
    \right.
\\& \left.
    - \frac{
            \xi^a_{, \theta} \xi^b_{, \theta}
        }{
            r^2 -2 m r + \sigma^2
        }
    \right]
\text{.}
\end{aligned}
\end{equation}
Therefore,
\begin{align}
    f  ( r, \theta )
&   = \mathscr{F} \rho^\mathfrak{a}
    \tan^\mathfrak{c} \frac{\theta}{2}
    \frac{\Delta^{\mathfrak{d} - \mathfrak{f}} \sin^{ 2\mathfrak{f} } \theta }{ ( r^2 - 2 m r + \sigma^2 )^\mathfrak{d} }
    \left( \frac{r_-}{r_+} \right)^\mathfrak{b}
    \left( \frac{R_-}{R_+} \right)^\mathfrak{e}
&   \text{for } m^2 > \sigma^2 ,
\\  f  ( r, \theta )
&   = \mathscr{F} \rho^\mathfrak{a}
    \tan^\mathfrak{c} \frac{\theta}{2}
    \frac{\Delta^{\mathfrak{d} - \mathfrak{f}} \sin^{ 2\mathfrak{f} } \theta }{ ( r^2 - 2 m r + \sigma^2 )^\mathfrak{d} }
    e^{ \mathfrak{b} \Psi + \mathfrak{e} \Theta }
&   \text{for } m^2 < \sigma^2 ,
\\  f  ( r, \theta )
&   = \mathscr{F} \rho^\mathfrak{a}
    \tan^\mathfrak{c} \frac{\theta}{2}
    \frac{\sin^\mathfrak{f} \theta}{( r - m )^\mathfrak{f}}
    \exp \left(
        \frac{ \mathfrak{b} + \mathfrak{e} \cos \theta }{r - m}
        + \frac{\mathfrak{d} \sin^2 \theta}{( r - m )^2}
    \right)
&   \text{for } m^2 = \sigma^2 ,
\end{align}
where $\Theta ( r, \theta ) = \tan^{-1} \frac{r - m}{\sqrt{\sigma^2 - m^2} \cos \theta}$,
$R_\pm = r - m \pm \sqrt{m^2 - \sigma^2} \cos \theta$, $\Delta = ( r - m )^2 + ( \sigma^2 - m^2 ) \cos^2 \theta$,
$\mathscr{F}$ is a positive constant,
$\mathfrak{a} = \frac{ \mathscr{C}^a \mathscr{C}^b }{4} \operatorname{tr} A_a A_b - 1 + \frac{1}{n}$,
$\mathfrak{c} = \frac{ \mathscr{C}^a \mathscr{E}^b }{2} \operatorname{tr} A_a A_b$,
$\mathfrak{f} = \frac{ \mathscr{E}^a \mathscr{E}^b }{8}\operatorname{tr} A_a A_b$,
\begin{align*}
    \mathfrak{b}
&   = \frac{\mathscr{C}^a \mathscr{D}^b}{4 \sqrt{ m^2 - \sigma^2 }} \operatorname{tr} A_a A_b,
&   \mathfrak{d}
&   = -\frac{\mathscr{D}^a \mathscr{D}^b}{8 ( m^2 - \sigma^2 )} \operatorname{tr} A_a A_b,
&   \mathfrak{e}
&   = \frac{\mathscr{D}^a \mathscr{E}^b}{4 \sqrt{ m^2 - \sigma^2 }} \operatorname{tr} A_a A_b
&   \text{if } m^2 > \sigma^2,
\\  \mathfrak{b}
&   = \frac{\mathscr{C}^a \mathscr{D}^b}{2 \sqrt{ \sigma^2 - m^2 }} \operatorname{tr} A_a A_b,
&   \mathfrak{d}
&   = \frac{\mathscr{D}^a \mathscr{D}^b}{8 ( \sigma^2 - m^2 )} \operatorname{tr} A_a A_b,
&   \mathfrak{e}
&   = \frac{\mathscr{D}^a \mathscr{E}^b}{2 \sqrt{ \sigma^2 - m^2 }} \operatorname{tr} A_a A_b
&   \text{if } m^2 < \sigma^2,
\\  \mathfrak{b}
&   = -\frac{\mathscr{C}^a \mathscr{D}^b}{2} \operatorname{tr} A_a A_b,
&   \mathfrak{d}
&   = -\frac{\mathscr{D}^a \mathscr{D}^b}{8} \operatorname{tr} A_a A_b,
&   \mathfrak{e}
&   = -\frac{\mathscr{D}^a \mathscr{E}^b}{2} \operatorname{tr} A_a A_b
&   \text{if } m^2 = \sigma^2.
\end{align*}

\begin{example}
\label{example: m2 geq}
Let $A_1 = \operatorname{diag} [ I_{n - 2}, 1, 1 - n ]$ and $A_2 = \operatorname{diag} [ I_{n - 2}, 2 - n, 0 ]$ be a pair of commuting matrices of $\mathfrak{sl} ( n, \mathbb{R} )$ with $n > 2$.
\end{example}
\noindent
We use Lemma \ref{lemma: exp matrices S} to determine $e^{\xi^a A_a} = \operatorname{diag} [ e^{\xi^1 + \xi^2} I_{n - 2}, e^{\xi^1 - (n - 2) \xi^2}, e^{-(n - 1) \xi^1} ]$.
Since $g_0 \in \mathcal{I} ( A_1, A_2 )$, by Lemma \ref{lemma: subspace ia matrices S}, the matrix constant $g_0$ has the form $g_0 = -\operatorname{diag} [ C, C_\phi, C_5 ]$, where $C_\phi, C_5 \in \mathbb{R}$ and $C \in \mathbf{Sym}_{n-2}$.
To obtain a solution to the chiral equation, we choose the solution \eqref{par sol m2 geq} for $\xi^a$, setting $\mathscr{D}^1 = 2 \sqrt{ m^2 - \sigma^2 } \frac{ p_3 }{ n - 1 }$, $\mathscr{D}^2 = 2 \sqrt{ m^2 - \sigma^2 } \frac{ p - q }{ n - 1 }$, $\mathscr{C}^1 = \frac{2}{n ( n - 1 )}$, $\mathscr{C}^2 = -\frac{2}{n - 1}$, $\mathscr{E}^1 = \mathscr{E}^2 = \xi^1_0 = \xi^2_0 = 0$, where $p_3 = (n - 2) p_1 + p_2$.
Then, we make the transformations $A_a \to P A_a P^{-1}$ and $g_0 \to P g_0 P^T$, where
\begin{equation}
\label{permutation matrix}
    P = \left[\begin{array}{rrr}
        & 1 & \\
        I_{n - 2} && \\
        && 1
    \end{array}\right] .
\end{equation}
Hence,
\begin{equation}
    g ( r, \theta ) = -\rho^{-\frac{2}{n} } \operatorname{diag} \left[
        C_\phi \rho^2 \left( \frac{r_-}{r_+} \right)^{p_2},
        \left( \frac{r_-}{r_+} \right)^{p_1} C,
        C_5 \left( \frac{r_+}{r_-} \right)^{p_3} 
    \right].
\end{equation}

In order to find exact solutions to the EFE, we set $\sigma = 0$ and $m > 0$.
If $n > 3$, we partition $C$ as
\begin{equation}
    C = \left[\begin{array}{ll}
        L_1 & L_2 \\
        L_2^T & L_3 \\ 
    \end{array}\right] ,
\end{equation}
where $L_1 \in \mathbb{R}$, $L_2 \in \mathbf{M}_{1, n - 3}$ and $L_3 \in \mathbf{Sym}_{n - 3}$.
Decomposing $g ( r, \theta )$ as in Eq. \eqref{decomposition g}, we obtain 
\begin{align}
\label{sol matrices S n3}
    \hat g_5
&   = g_{t t} dt^2 
    + g_{r r} dr^2
    + g_{\theta \theta} d\theta^2
    + g_{\phi \phi} d\phi^2
    + h_{5 5} ( dx^5 )^2
&   \text{ for } n = 3,
\\  \hat g
&   = \hat g_5
    + h_{i j} ( \omega^i + A_t^i dt ) ( \omega^j + A_t^j dt )
&   \text{ for } n > 3,
\end{align}
where
\begin{equation}
\begin{aligned}    
    g_{t t}
&   = -\left( 1 - \frac{2 m}{r} \right)^{p_1},
&   g_{r r}
&   = \left( 1 + \frac{m^2 \sin^2 \theta}{r ( r - 2 m )} \right)^l
    \left( 1 - \frac{2 m}{r} \right)^{p_2} ,
\\  g_{\phi \phi}
&   = \left( 1 - \frac{2 m}{r} \right)^{ p_2 + 1 }
    r^2 \sin^2 \theta ,
&   g_{\theta \theta}
&   = \left( 1 + \frac{m^2 \sin^2 \theta}{r ( r - 2 m )} \right)^l
    \left( 1 - \frac{2 m}{r} \right)^{ p_2 + 1 }
    r^2
\\  h_{5 5}
&   = \left(
        1
        - \frac{ 2 m }{ r }
    \right)^{-p_3},
&   h_{i j}
&   = \left( 1 - \frac{2 m}{r} \right)^{p_1} L_3,
\end{aligned}
\end{equation}
$A_t^i = L_3^{-1} L_2^T$ for all $i, j \in \{ 6, \ldots, n + 2 \}$ and $l = 1 - \frac{n^2 - 3 n + 2}{2} p_1^2 - ( n - 2 ) p_1 p_2 - p_2^2$.
The constant matrix $L_3$ also belongs to $SL ( n - 3, \mathbb{R} )$.
Since $L_2$ does not have restrictions, $A_t^i$ are arbitrary constants.

If $p = 1$ and $q = -1$, then $l = -\frac{( n - 2 )( n - 3 )}{2}$.
In particular, for $n = 3$, we find $l = 0$, and the metric tensor reduces to
\begin{equation}
    \hat g
    = -\left( 1 - \frac{2 m}{r} \right) dt^2 
    + \frac{ dr^2 }{ 1 - \frac{2 m}{r} }
    + r^2 ( d\theta^2 + \sin^2 \theta d\phi^2 )
    + ( dx^5 )^2.
\end{equation}
For $r >> m$, the components of the metric tensor have the following asymptotic expansions:
\begin{equation}
\begin{aligned}
    g_{t t}
&   = -1
    + 2 p_1 \frac{m}{r}
    - g_{t t}^{(2)} \frac{m^2}{r^2}
    + \ldots,
&   g_{r r}
&   = 1
    - 2 p_2 \frac{m}{r}
    + g_{r r}^{(2)} \frac{m^2}{r^2}
    + \ldots,
\\  g_{\theta \theta}
&   = \left(
        1
        - g_{\theta \theta}^{(1)} \frac{m}{r}
        + g_{\theta \theta}^{(2)} \frac{m^2}{r^2}
        + \ldots
    \right) r^2,
&   g_{\phi \phi}
&   = \left(
        1
        - g_{\theta \theta}^{(1)} \frac{m}{r}
        + p_2 g_{\theta \theta}^{(2)} \frac{m^2}{r^2}
        + \ldots
    \right) r^2 \sin^2 \theta,
\\  h_{5 5}
&   = 1
    + 2 l \frac{m}{r}
   + 2 l ( l + 1 ) \frac{m^2}{r^2}
    + \ldots,
&   h_{i j}
&   = \left(
        p_1
        - p_1 \frac{m}{r}
        + \frac{p_1 ( p_1 - 2 )}{2} \frac{m^2}{r^2}
        + \ldots
    \right) C,
\end{aligned}
\end{equation}
where
$g_{t t}^{(2)} = 2 p_1 ( p_1 - 1 )$,
$g_{r r}^{(2)} = 2 p_2 ( p_2 - 1 ) + l \sin^2 \theta$,
$g_{\theta \theta}^{(1)} = 2 ( p_2 + 1 )$ and
$g_{\theta \theta}^{(2)} = 2 p_2 ( p_2 + 1 ) + l \sin^2 \theta$.

\begin{example}
\label{example: m2 leq}
Let
\begin{equation}
    \left\{
        A_1 = \left[\begin{array}{rrr}
            I_{n - 2} && \\
            & 1 & \\
            && 1 - n \\ 
        \end{array}\right] ,
        A_2 = \left[\begin{array}{rrr}
            I_{n - 2} && \\
            & 2 - n & \\
            && 0
        \end{array}\right] ,
        A_3 = \left[\begin{array}{rrr}
            & -1 & \\
            1 && \\
            && 0_{n - 2}
        \end{array}\right]
    \right\}
\end{equation}
be a subset of pairwise commuting matrices of $\mathfrak{sl} ( n, \mathbb{R} )$ with $n > 3$.
\end{example}
\noindent By Lemmas \ref{lemma: g Lambda-alg 1st} and \ref{lemma: exp matrices S}, the exponential of $\xi^a A_a$ is
\begin{equation}
    e^{ \xi^a A_a } = \left[
        e^{ \xi^1 + \xi^2 }
        \left[\begin{array}{rrrr}
            \cos \xi^3 & -\sin \xi^3 \\
            \sin \xi^3 & \cos \xi^3 \\
        \end{array}\right],
        e^{ \xi^1 + \xi^2 } I_{n - 4},
        e^{ \xi^1 - ( n - 2 ) \xi^2 },
        e^{ -( n - 1 ) \xi^1 }
    \right].
\end{equation}
If $g_0$ is a constant matrix in $\mathcal{I} ( A_1, A_2, A_3 )$, by Theorem \ref{theo: ia jordan block 1st Lambda} and Lemma \ref{lemma: subspace ia matrices S}, $g_0$ has the form
\begin{equation}
    g_0 = - \operatorname{diag} \left[
        \left[\begin{array}{rr}
            C_t & D_t \\
            D_t & -C_t \\
        \end{array}\right],
        C,
        C_\phi,
        C_6
    \right],
\end{equation}
where $C \in \mathbf{Sym}_{n - 4}$ and $C_t, D_t, C_\phi, C_6 \in \mathbb{R}$.
To find a solution to the chiral equation \eqref{chiral eq g}, we make the transformations $A_a \to P A_a P^{-1}$ and $g_0 \to P g_0 P^T$, with $P$ given by Eq. \eqref{permutation matrix}, and choose the solution \eqref{par sol m2 leq} for $\xi^a$, fixing
$m = 0$,
$\mathscr{C}^1 = \frac{2}{n (n - 1)}$,
$\mathscr{C}^2 = -\frac{2}{n - 1}$,
$\mathscr{D}^1 = \sqrt{ \sigma^2 - m^2 } \frac{q_3}{n - 1}$,
$\mathscr{D}^2 = \sqrt{ \sigma^2 - m^2 } \frac{q_1 - q_2}{n - 1}$,
$\mathscr{D}^3 = \sqrt{ \sigma^2 - m^2 } p$,
$\xi^1_0 = -\frac{\pi}{2} \frac{q_3}{n - 1}$,
$\xi^2_0 = -\frac{\pi}{2} \frac{q_1 - q_2}{n - 1}$,
$\mathscr{C}^3 = \mathscr{E}^1 = \mathscr{E}^2 = \mathscr{E}^3 = \xi^3_0 = 0$,
where $q_3 = ( n - 2 ) q_1 + q_2$.
Hence,
\begin{equation}
\label{expA complex eigenvalues}
    g ( r, \theta ) = -\rho^{ -\frac{2}{n} } \left[
        C_\phi \rho^2 \frac{ e^{ q_2 \Psi } }{ e^{ q_2 \frac{\pi}{2} } },
        \frac{ e^{ q_1 \Psi } }{ e^{ q_1 \frac{\pi}{2} } }
        \left[\begin{array}{rrrr}
            U (r) & V (r) \\
            V (r) & -U (r) \\
        \end{array}\right],
        \frac{ e^{ q_1 \Psi } }{ e^{ q_1 \frac{\pi}{2} } } C,
        C_6 \frac{ e^{ q_3 \frac{\pi}{2} } }{ e^{ q_3 \Psi } }
    \right],
\end{equation}
where $U (r) = C_t \cos p \Psi - D_t \sin p \Psi$ and $V (r) = D_t \cos p \Psi + C_t \sin p \Psi$.

The mathematical expression of the trigonometric functions $U (r)$ and $V (r)$ depends on the type of number that $p$ is.
When $p$ is a natural number, $\cos p \Psi$ and $\sin p \Psi$ are related to Chebyshev polynomials of the first and second kind, $T_p$ and $U_p$, respectively, as follows: $\cos p \Psi = T_p ( \cos \Psi )$ and $\sin p \Psi = \sin \Psi U_p ( \cos \Psi )$.
It is convenient to define the functions
\begin{align}
    \mathscr{U}_p ( x )
&   = \mu T_p \left( \frac{1}{\sqrt{ x^2 + 1 }} \right)
    - \frac{\nu x}{\sqrt{ x^2 + 1 }} U_{p - 1} \left( \frac{1}{\sqrt{ x^2 + 1 }} \right),
\\  \mathscr{V}_p ( x )
&   = \nu T_p \left( \frac{1}{\sqrt{ x^2 + 1 }} \right)
    + \frac{\mu x}{\sqrt{ x^2 + 1 }} U_{p - 1} \left( \frac{1}{\sqrt{ x^2 + 1 }} \right)
\end{align}
for all $x \in \mathbb{R}$, where $\mu$ and $\nu$ are constant parameters.
Using the recurrence relations of the Chebyshev polynomials, we find
\begin{equation}
    \mathscr{U}_p ( x )
    = \frac{2}{\sqrt{ x^2 + 1 }} \mathscr{U}_{p - 1} ( x )
    - \mathscr{U}_{p - 2} ( x )
    \text{ and }
    \mathscr{V}_p ( x )
    = \frac{2}{\sqrt{ x^2 + 1 }} \mathscr{V}_{p - 1} ( x )
    - \mathscr{V}_{p - 2} ( x )
\end{equation}
for all $p > 2$.
Some properties of these functions are: $\mathscr{U}_p^2 ( x ) + \mathscr{V}_p^2 ( x ) = \mu^2 + \nu^2$, $\mathscr{U}_p ( 0 ) = \mu$ and $\mathscr{V}_p ( 0 ) = \nu$, $\frac{d \mathscr{U}_p}{d x} ( x ) = - \frac{p}{ x^2 + 1 } \mathscr{V}_p ( x )$ and $\frac{d \mathscr{V}_p}{d x} ( x ) = \frac{p}{ x^2 + 1 } \mathscr{U}_p ( x )$,
\begin{align*}
    \lim _{x \to \pm \infty} \mathscr{U}_p ( x )
&   = (-)^\frac{p}{2} \mu,
&   \lim _{x \to \pm \infty} \mathscr{V}_p ( x )
&   = (-)^\frac{p}{2} \nu
&   \text{for $p$ even,} 
\\  \lim _{x \to \pm \infty} \mathscr{U}_p ( x )
&   = \mp (-)^\frac{p - 1}{2} \nu,
&   \lim _{x \to \pm \infty} \mathscr{V}_p ( x )
&   = \pm (-)^\frac{p - 1}{2} \mu
&   \text{for $p$ odd.} 
\end{align*}
Let $k \in \{ 0, \ldots, p - 1\}$.
The roots of $\mathscr{U}_p$ and $\mathscr{V}_p$ are given by
\begin{equation}
    x_k = \tan \frac{
        \tan^{-1} \mu/\nu + k \pi
    }{p} \text{ and }
    x_k = \tan \frac{
        -\tan^{-1} \nu/\mu + k \pi
    }{p} \text{, respectively.}
\end{equation}

In order to obtain an exact solution to the EFE, we set $m = 0$ and assume that $p \in \mathbb{N}$ and $\sigma > 0$. 
Then, we decompose the matrix $g ( r, \theta )$ as in Eq. \eqref{decomposition g}.
Therefore,
\begin{align}
    \hat g_6
&   = g_{t t} dt^2
    + g_{r r} dr^2
    + g_{\theta \theta} d\theta^2
    + g_{\phi \phi} d\phi^2
    + h_{5 5} ( dx^5 + A_t^5 dt )^2
    + h_{6 6} ( dx^6 )^2
&   \text{for } n = 4,
\\  \hat g
&   = \hat g_6
    + h_{i j} dx^i dx^j
&   \text{for } n > 4,
\end{align}
where
\begin{equation}
\begin{aligned}
    g_{t t}
&   = - \frac{ \mathscr{U}_\infty }{ e^{ q_1 \frac{\pi}{2} } }
    \frac{ e^{ q_1 \Psi ( r ) } }{ \mathscr{U}_p ( r/\sigma ) },
&   g_{r r}
&= \left(
        1
        - \frac{ \sigma^2 \sin^2 \theta }{ r^2 + \sigma^2 }
    \right)^l
    \frac{ e^{ q_2 \Psi ( r ) } }{ e^{ q_2 \frac{\pi}{2} } },
\\  g_{\theta \theta}
&   = \left(
        1
        - \frac{ \sigma^2 \sin^2 \theta }{ r^2 + \sigma^2 }
    \right)^l
    \frac{ e^{ q_2 \Psi ( r ) } }{ e^{ q_2 \frac{\pi}{2} } }
    ( r^2 + \sigma^2 ),
&   g_{\phi \phi}
&   = \frac{ e^{ q_2 \Psi ( r ) } }{ e^{ q_2 \frac{\pi}{2} } }
    ( r^2 + \sigma^2 )
    \sin^2 \theta,
\\  h_{5 5}
&   = - \mathscr{U}_p ( r/\sigma )
    \frac{ e^{ q_1 \Psi ( r ) } }{ e^{ q_1 \frac{\pi}{2} } },
&   A_t^5
&   = - \frac{ \mathscr{V}_p ( r/\sigma ) }{ \mathscr{U}_p ( r/\sigma ) },
\\  h_{6 6}
&   = - \frac{1}{ \mathscr{U}_\infty } \frac{ e^{ q_3 \frac{\pi}{2} } }{ e^{ q_3 \Psi ( r ) } },
&   h_{i j}
&   = \frac{ e^{ q_1 \Psi ( r ) } }{ e^{ q_1 \frac{\pi}{2} } } C
\end{aligned}
\end{equation}
for all $i, j \in \{ 7, \ldots, n + 2 \}$, where
$\mu = C_t$,
$\nu = D_t$,
$8 l = (n - 1) (n - 2) q_1^2 + 2 (n - 2 ) q_1 q_2 + 2 q_2^2  - 2 p^2 + 8$,
$\mathscr{U}_\infty = \lim _{x \to \infty} \mathscr{U}_p ( x ) $.
The constant parameters $\mu$ and $\nu$ satisfy $\mu^2 + \nu^2 = - \mathscr{U}_\infty$.
$C$ is also a constant matrix in $SL ( n - 4, \mathbb{R} )$.

The asymptotic expansions of $\mathscr{U}_p ( x )$ and $\mathscr{V}_p ( x )$ for $x >> 1$ are
\begin{equation}
    \mathscr{U}_p ( x )
    = \mathscr{U}_\infty
    + \frac{ p \mathscr{V}_\infty }{x}
    - \frac{ p^2 \mathscr{U}_\infty }{2 x^2}
    + \ldots \text{ and }
    \mathscr{V}_p ( x )
    = \mathscr{V}_\infty
    - \frac{ p \mathscr{U}_\infty }{x}
    - \frac{ p^2 \mathscr{V}_\infty }{ 2 x^2 }
    + \ldots
\end{equation}
Using them, we find
\begin{equation}
\begin{aligned}
    g_{t t}
&   = -1
    - ( p A_\infty - q_1 ) \frac{\sigma}{r}
    - \frac{ g_{t t}^{(2)}}{2} \frac{\sigma^2}{r^2}
    + \ldots,
&   g_{r r}
&   = 1
    - q_2 \frac{\sigma}{r}
    + \frac{ g_{r r}^{(2)} }{2} \frac{\sigma^2}{r^2}
    + \ldots,
\\  g_{\theta \theta}
&   = r^2 \left(
        1
        - q_2 \frac{\sigma}{r} 
        + \frac{
            2 
            + g_{r r}^{(2)}
        }{2} \frac{\sigma^2}{r^2}
        + \ldots
    \right),
&   g_{\phi \phi}
&   = r^2 \sin^2 \theta \left(
        1
        - q_2 \frac{\sigma}{r} 
        + \frac{
            2 
            + q_2^2
        }{2} \frac{\sigma^2}{r^2}
        + \ldots
    \right),
\\  h_{5 5}
&   = -\mathscr{U}_\infty
    + h_{5 5}^{(1)} \frac{\sigma}{r}
    - h_{5 5}^{(2)} \frac{\sigma^2}{r^2}
    + \ldots,
&   A_t^5
&   = A_\infty
    + \frac{ p }{ \mathscr{U}_\infty } \frac{\sigma}{r}
    + A_t^{(2)} \frac{\sigma^2}{r^2}
    + \ldots,
\\  h_{6 6}
&   = - \frac{1}{\mathscr{U}_\infty}
    - \frac{q_3 }{\mathscr{U}_\infty} \frac{\sigma}{r}
    - \frac{q_3^2}{2 \mathscr{U}_\infty} \frac{\sigma^2}{r^2}
    + \ldots,
\end{aligned} 
\end{equation}
where
$\mathscr{V}_\infty = \lim _{x \to \infty} \mathscr{V}_p ( x ) $,
$A_\infty = - \frac{ \mathscr{V}_\infty }{ \mathscr{U}_\infty }$,
$g_{t t}^{(2)} = q_1^2 + p^2 + 2 p A_\infty ( p A_\infty - q_1 )$,
$g_{r r}^{(2)} = q_2^2 - 2 l \sin^2 \theta$,
$h_{5 5}^{(1)} = \mathscr{U}_\infty ( p A_\infty + q_1 )$,
$h_{5 5}^{(2)} = \mathscr{U}_\infty ( q_1^2 - p^2 +2 p q_1 A_\infty )$ and
$A_t^{(2)} = A_\infty ( p^2 + A_\infty^2 )$.

If $p = 2$, $\nu = 0$ and $C = I_{n - 3}$, we obtain
\begin{equation}
    \hat g
    = \frac{ r^2 + \sigma^2 }{ r^2 - \sigma^2 } dt^2
    + dr^2
    + ( r^2 + \sigma^2 ) d\Omega^2
    + \frac{ r^2 - \sigma^2 }{ r^2 + \sigma^2 } \left(
        dx^5
        + \frac{ 2 \sigma r dt }{ r^2 - \sigma^2 }
    \right)^2
    + ( dx^6 )^2
    + \ldots
    + ( dx^{n + 2} )^2.
\end{equation}
In Figure \ref{fig: gtt m2 leq}, we plot the behavior of $g_{tt} ( r )$ assuming $\nu \neq 0$.
Observe that the singular points of $g_{tt}$ are greater than 1 when $\nu$ is negative, while they are less than 1 for $\nu > 0$.
Also, $g_{tt}$ approaches more rapidly to -1 when $\nu$ is positive.

\begin{figure}[ht]
    \centering
    \begin{subfigure}{0.49\textwidth}
        \centering
        \includegraphics[width=1\textwidth]{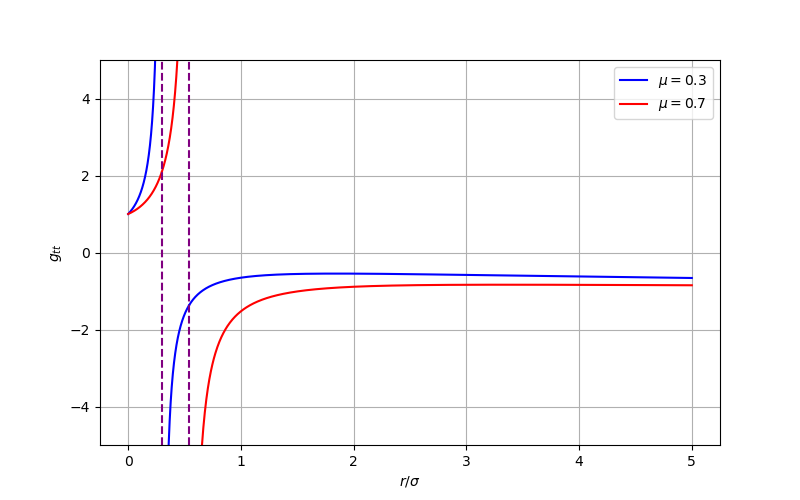}
        \caption{$\nu = 0.458258$.}
        \label{subfig: gtt m2 leq mu3 nupos}
    \end{subfigure}
    \hfill
    \begin{subfigure}{0.49\textwidth}
        \centering
        \includegraphics[width=1\textwidth]{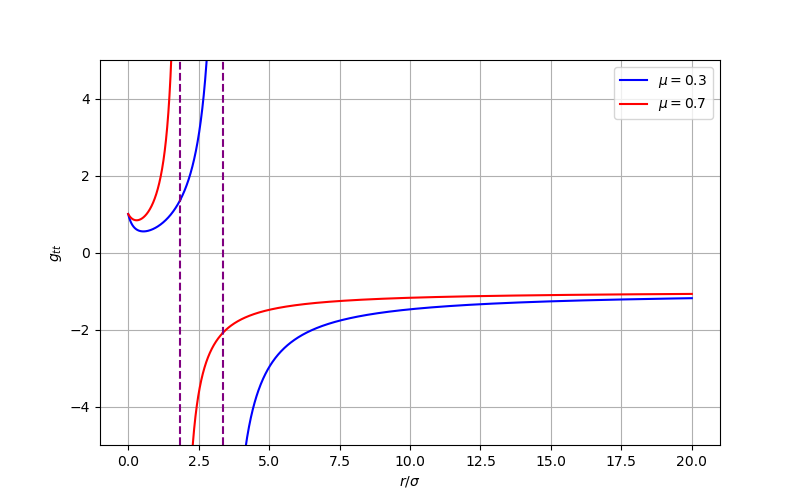}
        \caption{$\nu = -0.458258$.}
        \label{subfig: gtt m2 leq mu7 nupos}
    \end{subfigure}

    \caption{Behavior of $g_{tt} (r)$ with $p = 2$ and $q_1 = 0$.}
    \label{fig: gtt m2 leq}
\end{figure}

\begin{example}
\label{example: m2 eq}
Let
\begin{equation}
    \left\{
        A_1 = \left[\begin{array}{rrr}
            I_{n - 2} && \\
            & 1 & \\
            && 1 - n \\ 
        \end{array}\right] ,
        A_2 = \left[\begin{array}{rrr}
            I_{n - 2} && \\
            & 2 - n & \\
            && 0
        \end{array}\right] ,
        A_3 = \left[\begin{array}{rrr}
            0 & 1 & \\
            & 0 & \\
            && 0_{n - 2}
        \end{array}\right]
    \right\}
\end{equation}
be a subset of pairwise commuting matrices of $\mathfrak{sl} ( n, \mathbb{R} )$ with $n > 3$.
\end{example}
\noindent Using Lemmas \ref{lemma: g lambda-alg 1st} and \ref{lemma: exp matrices S}, we determine the exponential of the linear combination $\xi^a A_a$.
Then,
\begin{equation}
    e^{\xi^a A_a} = \operatorname{diag} \left[
        e^{\xi^1 + \xi^2}
        \left[\begin{array}{rr}
            1 & \xi^3 \\
            & 1 \\
        \end{array}\right],
        e^{\xi^1 + \xi^2} I_{n - 4},
        e^{\xi^1 - (n - 2 ) \xi^2},
        e^{- (n - 1 ) \xi^1}
    \right].
\end{equation}
By Theorem \ref{theo: ia jordan block 1st lambda}, any constant matrix $g_0$ that belongs to $\mathcal{I} ( A_1, A_2, A_3 )$ has the form
\begin{equation}
    g_0 = -\operatorname{diag} \left[
        \left[\begin{array}{rr}
            C_t & D_t \\
            D_t & \\
        \end{array}\right],
        C,
        C_\phi,
        C_6
    \right],
\end{equation}
where $C_\phi, C_t, D_t, C_6 \in \mathbb{R}$ and $C \in \mathbf{Sym}_{n - 4}$.
In this example, the solutions of Eq. \eqref{gen Laplace eq} are given by Eq. \eqref{par sol m2 eq}, with $\mathscr{C}^1 = \frac{2}{( n - 1 ) n}$, $\mathscr{C}^2 = -\frac{2}{n - 1}$, $\mathscr{D}^1 = \frac{q_3}{n - 1}$, $\mathscr{D}^2 = \frac{q_1 - q_2}{n - 1}$, $\mathscr{C}^3 = \mathscr{E}^1 = \mathscr{E}^2 = \mathscr{E}^3 = \xi^1 = \xi^2 = \xi^3 = 0$, where $q_1$ and $q_2$ are constant parameters, $q_3 = (n - 2) q_1 + q_2$.
Then, we make the transformations $A_a \to P A_a P^{-1}$ and $g_0 \to P g_0 P^T$, with
\begin{equation}
    P = \left[\begin{array}{rrrrr}
        &&& 1 & \\
        & 1 &&& \\
        1 &&&& \\
        && I_{n - 4} && \\
        &&&& 1 \\
    \end{array}\right].
\end{equation}
Thus,
\begin{equation}
    g ( r, \theta ) = -\rho^{-\frac{2}{n}} \operatorname{diag} \left[
        C_\phi \rho^2 e^\frac{-q_2}{r - m},
        e^\frac{-q_1}{r - m} \left[\begin{array}{cc}
            & D_t \\
            D_t & C_t - \frac{\mathscr{D}^3 D_t}{r - m} 
        \end{array}\right],
        e^\frac{-q_1}{r - m} C,
        C_6 e^{\frac{q_3}{r - m}}
    \right]
\end{equation}
is a solution of the chiral equation \eqref{chiral eq g}.

In order to obtain exact solutions to the EFE, we decompose $g ( r, \theta )$ as in Eq. \eqref{decomposition g} and set $m = - \frac{\mathscr{D}^3}{D_0} = 2 M$.
Therefore,
\begin{align}
\label{6 dim metric tensor m2 eq}
    \hat g_6
&   = g_{t t} dt^2
    + g_{r r} dr^2
    + g_{\theta \theta} d\theta^2
    + g_{\phi \phi} d\phi^2
    + h_{5 5} ( dx^5 + A_t^5 dt )^2
    + h_{6 6} ( dx^6 )^2
&   \text{for } n = 4,
\\\label{higher dim metric tensor m2 eq}
    \hat g
&   = \hat g_6
    + h_{i j} dx^i dx^j
&   \text{for } n > 4,
\end{align}
where
\begin{equation}    
\begin{aligned}
    g_{tt}
&   = -\left(
        1
        - \frac{2 M}{r}
    \right) e^\frac{-q_1}{r - 2 M},
&   g_{\theta\theta}
&   = \left(
        1
        - \frac{2 M}{r}
    \right)^2
    e^{
        -\frac{q_2}{r - 2 M}
        - \frac{l \sin^2 \theta}{( r - 2 M )^2}
    }
    r^2,
\\  g_{rr}
&   = e^{
        -\frac{q_2}{r - 2 M}
        - \frac{l \sin^2 \theta}{( r - 2 M )^2}
    },
&   g_{\phi\phi}
&   = \left(
        1
        - \frac{2 M}{r}
    \right)^2
    e^\frac{-q_2}{r - 2 M}
    r^2 \sin^2 \theta,
\\  h_{5 5}
&   = D_t^2 \left(
        1
        - \frac{2 M}{r}
    \right)^{-1} e^\frac{-q_1}{r - 2 M},
&   A_t^5
&   = \frac{1}{D_t} \left(
        1
        - \frac{2 M}{r}
    \right)
\\  h_{6 6}
&   = \frac{1}{D_t^2} e^\frac{q_3}{r - 2 M},
&   h_{i j}
&   = e^\frac{-q_1}{r - 2 M} C
\end{aligned}
\end{equation}
for all $i, j \in \{ 7, \ldots, n + 2 \}$, $8 l = ( n - 1 ) ( n - 2 ) q_1^2 + 2 ( n - 2 ) q_1 q_2 + 2 q_2^2$ and $C$ also belongs to $SL ( n - 4, \mathbb{R} )$.

The asymptotic behavior of the components of the metric tensors \eqref{6 dim metric tensor m2 eq} and \eqref{higher dim metric tensor m2 eq} for $r >> 2 M$ are
\begin{equation}    
\begin{aligned}
    g_{t t}
&   = -1
    + \frac{q_1 + 2 M}{r}
    - \frac{q_1^2}{2 r^2}
    + \ldots,
&   g_{r r}
&   = 1
    - \frac{q_2}{r}
    + \frac{ g_{r r}^{(1)} }{2 r^2}
    + \ldots,
\\  g_{\phi \phi}
&   = r^2 \sin^2 \theta \left(
        1
        - \frac{ q_2 + 4 M }{r}
    + \frac{ g_{\phi \phi}^{(2)} }{2 r^2}
    + \ldots
    \right),
&   g_{\theta \theta}
&   = r^2 \left(
        1
        - \frac{ q_2 + 4 M }{r}
    + \frac{ g_{\theta \theta}^{(2)} }{2 r^2}
    + \ldots
    \right),
\\  h_{5 5}
&   = D_t^2
    + \frac{D_t^2 ( 2 M - q_1 )}{r}
   + \frac{ h_{5 5}^{(2)} }{2 r^2}
    + \ldots,
&   h_{6 6}
&   = \frac{1}{ D_t^2 }
    + \frac{q_3}{D_t^2 r}
   + \frac{ h_{6 6}^{(2)} }{2 D_t^2 r^2}
    + \ldots,
\\  h_{i j}
&   = \left(
        1
        - \frac{q_1}{r}
        + \frac{ h_{i j}^{(2)} }{2 r^2}
        + \ldots
    \right) C,
\end{aligned}
\end{equation}
where $g_{r r}^{(1)} = q_2 ( q_2 - 4 M ) - 2 l \sin^2 \theta$,
$g_{\phi \phi}^{(2)} = q_2 ( q_2 + 4 M ) + 8 M^2$,
$g_{\theta \theta}^{(2)} = g_{\phi \phi}^{(2)} - 2l \sin^2\theta$,
$h_{5 5}^{(2)} = D_t^2 ( 8 M ( M - q_1 ) + q_1^2 )$,
$h_{6 6}^{(2)} = q_3 ( q_3 + 4 M )$ and
$h_{i j}^{(2)} = q_1 ( q_1 - 4 M )$.

For $q_1 = q_2 = 0$, $D_t = 1$ and $n = 4$, $\hat g_6$ reduces to
\begin{equation}
    \hat g_6
    = -\left(
        1
        - \frac{2 M}{r}
    \right) dt^2
    + dr^2
    + \left(
        1
        - \frac{2 M}{r}
    \right)^2 r^2 d\Omega^2
    + \frac{\left(
        dx^5
        + \left(
            1
            - \frac{2 M}{r}
        \right) dt
    \right)^2}{
        1
        - \frac{2 M}{r}
    }
    + (dx^6 )^2.
\end{equation}
In Figure \ref{fig: gtt_meq}, we plot $g_{tt} (r)$ for three values of the ratio $q_1/M$: 0, 0.5 and 1.5.
Observe that $g_{tt}$ tends to $0$ as $r$ approaches $2 M$ from the right for $q_1 > 0$.
$g_{tt}$ decreases slowly as $q_1/M$ increases.

\begin{figure}
    \centering
    \includegraphics[width=.5\textwidth]{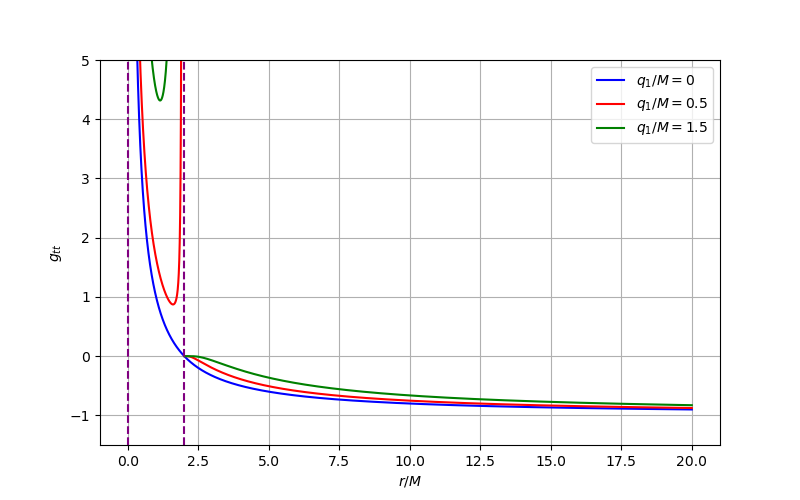}
    \caption{Evolution of $g_{tt} (r)$ considering three values of $q_1/M$: 0, 0.5 and 1.5.}
    \label{fig: gtt_meq}
\end{figure}

\section{Conclusion}    %
\label{section: conclusions}

In this work, we introduced a method for obtaining exact solutions to the higher-dimensional EFE in vacuum from a given solution to the chiral equation.
We solved the chiral equation for a symmetric matrix $g ( z, \bar{z} ) \in SL ( n, \mathbb{R} )$ assuming a flat space.
The matrix $g ( z, \bar{z} )$ is determined by a set $\{ A_a \} \subset \mathfrak{sl} ( n, \mathbb{R} )$ of pairwise commuting matrices, a constant matrix $g_0 \in \mathcal{I} \{ A_a \}$ and a set $\{ \xi^a ( z, \bar{z} ) \}$ of solutions to the generalized Laplace equation \eqref{gen Laplace eq}.
The advantage of this method lies in the use of algebraic techniques to determine the set $\{ A_a \}$ and the matrix $g_0$.

To determine a pair of commuting matrices, we assumed that $A_1$ be a representative of some equivalence class of $\mathfrak{sl} ( n, \mathbb{R} )$ and $A_2$ a matrix in the centralizer of $A_1$.
In Section \ref{section: centralizer}, we determined the centralizers of the five types of equivalence classes of $\mathfrak{sl} ( n, \mathbb{R} )$.

For a set of more than two matrices, we proposed to find the largest commutative algebras $\mathfrak{A}$ included in the centralizers of equivalence classes of $\mathfrak{sl} ( n, \mathbb{R} )$.
The matrices $A_a$ are chosen from these commutative algebras.
In Section \ref{section: commutative algebras}, we found the five types of commutative algebras. 
From the commutative algebra $\mathfrak{J}_{1, \ldots, 1} \subset \mathbf{M}_n$, we obtained the set $\{ S_0, \ldots, S_{n - 2} \} \in \mathfrak{sl} ( n, \mathbb{R} )$, which is linearly independent.

Since the matrix $g_0$ belongs to the subspace $\mathcal{I} \{ A_a \}$, in Section \ref{section: subspace ia} we presented this set for each of the five types of equivalence classes of $\mathfrak{sl} ( n, \mathbb{R} )$.
In addition, we determined $\mathcal{I} ( \mathfrak{A} )$ for the commutative algebras $\mathfrak{A}$ obtained in Section \ref{section: commutative algebras} and $\mathcal{I} ( S_0, \ldots, S_i )$ for some $i \in \{ 1, \ldots, n - 2 \}$.

Everyone knows that calculating the exponential of a matrix is a long and complicated procedure.
That is why, in Section \ref{section: matrix exponential}, we determined the exponential of matrices in $\mathfrak{J}_m$ and $\mathfrak{J}^n$.
We also computed the exponential of a linear combination of matrices $S_0, \ldots, S_i$ for some $i \in \{ 1, \ldots, n - 2 \}$.

As examples, we built three metric tensors using our method.
These metric tensors depend on constant parameters.
Varying these parameters yields different solutions to the EFE.
The metrics of the four-dimensional base spaces are asymptotically flat.
For these examples, we only used the matrices $S_0$ and $S_1$; however, it is possible to obtain more solutions by using the other matrices $S_3, \ldots, S_{n - 1}$.

In summary, we introduced a new and simple method, along with the necessary tools to obtain exact solutions to the EFE in vacuum for any dimension greater than 4.
Furthermore, we showed how to write these metrics as a Kaluza-Klein metric.

\backmatter

\bmhead{Acknowledgements}

This work was partially supported by CONAHCyT M\'exico under grants \mbox{A1-S-8742}, 304001, 376127, 240512,  \mbox{FORDECYT-PRONACES} grant No. 490769 and \mbox{I0101/131/07 C-234/07} of the Instituto Avanzado de Cosmolog\'ia (IAC) collaboration (http://www.iac.edu.mx/).

\begin{appendices}

\section{Equivalence classes}   %

\begin{longtable}{ c c c c}
    $\mathbf{Class}$    &   $\mathbf{A}_1$    &   $\mathbf{A}_2$  &   $\mathbf{g}_0$  \\\hline\hline
    $[ A_1 ]$   &   $\left[\begin{array}{ccc}
        0 & 1 & \\
        & 0 & \\
        && 0 \\
    \end{array}\right]$ & $\left[\begin{array}{ccc}
        a_1 & a_2 & c \\
        & a_1 & \\
        & d & b   \\
    \end{array}\right]$  &   $\left[\begin{array}{ccc}
        A_1 & A_2 & C \\
        A_2 && \\
        C && B   \\
    \end{array}\right]$  \\\hline
    $[ A_2 ]$   &  $\left[\begin{array}{ccc}
        r && \\
        & r &\\
        && -2 r \\
    \end{array}\right]$ &   $\left[\begin{array}{ccc}
        a_{11} & a_{12} & \\
        a_{21} & a_{22} & \\
        && b \\
    \end{array}\right]$ &   $\left[\begin{array}{ccc}
        A_1 & A_2 & \\
        A_2 & A_3 & \\
        && B \\
    \end{array}\right]$ \\\hline
    $[ A_3 ]$   &   $\left[\begin{array}{ccc}
        0 & 1 & 0   \\
        & 0 & 1 \\
        && 0 \\
    \end{array}\right]$  &   $\left[\begin{array}{ccc}
        a_1 & a_2 & a_3   \\
        & a_1 & a_2 \\
        && a_1 \\
    \end{array}\right]$ &   $\left[\begin{array}{ccc}
        A_1 & A_2 & A_3 \\
        A_2 & A_3 & \\
        A_3 && \\
    \end{array}\right]$   \\\hline
    $[ A_4 ]$   &   $\left[\begin{array}{ccc}
        r & 1 & \\
        & r & \\
        && -2 r \\
    \end{array}\right]$ &   $\left[\begin{array}{ccc}
        a_1 & a_2 & \\
        & a_1 & \\
        && b \\
    \end{array}\right]$ &   $\left[\begin{array}{ccc}
        A_1 & A_2 & \\
        A_2 && \\
        && B \\
    \end{array}\right]$\\\hline
    $[ A_5 ]$   &   $\left[\begin{array}{ccc}
        r_1 && \\
        & r_2 & \\
        && r_3 \\
    \end{array}\right]$ &   $\left[\begin{array}{ccc}
        a && \\
        & b & \\
        && c \\
    \end{array}\right]$ &   $\left[\begin{array}{ccc}
        A && \\
        & B & \\
        && C \\
    \end{array}\right]$ \\\hline
    $[ A_6 ]$   &   $\left[\begin{array}{ccc}
        -2 r && \\
        & r & -s \\
        & s & r \\
    \end{array}\right]$ &   $\left[\begin{array}{ccc}
        a && \\
        & u & -v \\
        & v & u \\
    \end{array}\right]$ &   $\left[\begin{array}{ccc}
        A && \\
        & U & V \\
        & V & -U \\
    \end{array}\right]$\\\hline
\caption{Equivalence classes of $\mathfrak{sl} ( 3, \mathbb{R} )$.}
\label{table: eq classes sl3}
\end{longtable}

\begin{longtable}{ c c c c }
    $\mathbf{Class}$    &   $\mathbf{A}_1$    &   $\mathbf{A}_2$  &   $\mathbf{g}_0$  \\\hline\hline
    $[ A_1 ]$   &   $\left[\begin{array}{cccc}
        r &&& \\
        & r && \\
        && -r & \\
        &&& -r \\
    \end{array}\right]$ &   $\left[\begin{array}{cccc}
        a_{11} & a_{12} && \\
        a_{21} & a_{22} && \\
        && b_{11} & b_{12} \\
        && b_{21} & b_{22} \\
    \end{array}\right]$ &   $\left[\begin{array}{cccc}
        A_1 & A_2 && \\
        A_2 & A_3 && \\
        && B_1 & B_2 \\
        && B_2 & B_3 \\
    \end{array}\right]$ \\\hline
    $[ A_2 ]$   &   $\left[\begin{array}{cccc}
        0 & 1 && \\
        & 0 && \\
        && 0 & \\
        &&& 0 \\
    \end{array}\right]$ &   $\left[\begin{array}{cccc}
        a_1 & a_2 & c_1 & c_2 \\
        & a_1 && \\
        & d_1 & b_{11} & b_{12} \\
        & d_2 & b_{21} & b_{22} \\
    \end{array}\right]$  &   $\left[\begin{array}{cccc}
        A_1 & A_2 & C_1 & \\
        A_2 & A_3 & C_2 & \\
        C_1 & C_2 & B_1 & B_2 \\
        && B_2
    \end{array}\right]$  \\\hline
    $[ A_3 ]$   &   $\left[\begin{array}{cccc}
        r & 1 && \\
        & r && \\
        && -r & \\
        &&& -r \\
    \end{array}\right]$ &   $\left[\begin{array}{cccc}
        a_1 & a_2 && \\
        & a_1 && \\
        && b_{11} & b_{12} \\
        && b_{21} & b_{22}
    \end{array}\right]$  &   $\left[\begin{array}{cccc}
        A_1 & A_2 && \\
        A_2 &&& \\
        && B_1 & B_2 \\
        && B_2 & B_3
    \end{array}\right]$  \\\hline
    $[ A_4 ]$   &   $\left[\begin{array}{cccc}
        r &&& \\
        & r && \\
        && -r & -s \\
        && s & -r \\
    \end{array}\right]$ &   $\left[\begin{array}{cccc}
        a_{11} & a_{12} && \\
        a_{21} & a_{22} && \\
        && u & -v \\
        && v & u \\
    \end{array}\right]$ &   $\left[\begin{array}{cccc}
        A_1 & A_2 && \\
        A_2 & A_3 && \\
        && U & V \\
        && V & -U \\
    \end{array}\right]$ \\\hline
    $[ A_5 ]$   &   $\left[\begin{array}{cccc}
        0 & 1 && \\
        & 0 && \\
        && 0 & 1 \\
        &&& 0 \\
    \end{array}\right]$ &   $\left[\begin{array}{cccc}
        a_1 & a_2 & c_1 & c_2 \\
        & a_1 && c_1 \\
        d_1 & d_2 & b_1 & b_2 \\
        & d_1 && b_1
    \end{array}\right]$ &   $\left[\begin{array}{cccc}
        A_1 & A_2 & C_1 & C_2 \\
        A_2 && C_1 \\
        C_1 & C_2 & B_1 & B_2 \\
        C_2 && B_2 &
    \end{array}\right]$ \\\hline
    $[ A_6 ]$   &   $\left[\begin{array}{cccc}
        r & 1 && \\
        & r && \\
        && -r & 1 \\
        &&& -r \\
    \end{array}\right]$ &   $\left[\begin{array}{cccc}
        a_1 & a_2 && \\
        & a_1 && \\
        && b_1 & b_2 \\
        &&& b_1 \\
    \end{array}\right]$ &   $\left[\begin{array}{cccc}
        A_1 & A_2 && \\
        A_2 &&& \\
        && B_1 & B_2 \\
        && B_2 & \\
    \end{array}\right]$ \\\hline
    $[ A_7 ]$   &   $\left[\begin{array}{cccc}
        r & 1 && \\
        & r && \\
        && -r & -s \\
        && s & -r \\
    \end{array}\right]$ &   $\left[\begin{array}{cccc}
        a_1 & a_2 && \\
        & a_1 && \\
        && u & -v \\
        && v & u \\
    \end{array}\right]$ &   $\left[\begin{array}{cccc}
        A_1 & A_2 && \\
        A_2 &&& \\
        && U & V \\
        && V & -U \\
    \end{array}\right]$ \\\hline
    $[ A_8 ]$   &   $\left[\begin{array}{cccc}
        r & -s_1 && \\
        s_1 & r && \\
        && -r & -s_2 \\
        && s_2 & -r \\
    \end{array}\right]$ &   $\left[\begin{array}{cccc}
        u_1 & -v_1 && \\
        v_1 & u_1 && \\
        && u_2 & -v_2 \\
        && v_2 & u_2 \\
    \end{array}\right]$ &   $\left[\begin{array}{cccc}
        U_1 & V_1 && \\
        V_1 & -U_1 && \\
        && U_2 & V_2 \\
        && V_2 & -U_2 \\
    \end{array}\right]$ \\\hline
    $[ A_9 ]$   &   $\left[\begin{array}{cccc}
        0 & -s && \\
        s & 0 && \\
        && 0 & -s \\
        && s & 0 \\
    \end{array}\right]$ &   $\left[\begin{array}{cccc}
        u_{11} & -v_{11} & u_{12} & -v_{12} \\
        v_{11} & u_{11} & v_{12} & u_{12} \\
        u_{21} & -v_{21} & u_{22} & -v_{22} \\
        v_{21} & u_{21} & v_{22} & u_{22}
    \end{array}\right]$ &   $\left[\begin{array}{cccc}
        U_1 & V_1 & U_2 & V_2 \\
        V_1 & -U_1 & V_2 & -U_2 \\
        U_2 & V_2 & U_3 & V_3 \\
        V_2 & -U_2 & V_3 & -U_3
    \end{array}\right]$ \\\hline
    $[ A_{10} ]$  &   $\left[\begin{array}{cccc}
        r_1 &&& \\
        & r_1 && \\
        && r_2 & \\
        &&& r_3 \\
    \end{array}\right]$ &   $\left[\begin{array}{cccc}
        a_{11} & a_{12} && \\
        a_{21} & a_{22} && \\
        && b & \\
        &&& c \\
    \end{array}\right]$ &   $\left[\begin{array}{cccc}
        A_1 & A_2 && \\
        A_2 & A_3 && \\
        && B & \\
        &&& C \\
    \end{array}\right]$ \\\hline
    $[ A_{11} ]$  &   $\left[\begin{array}{cccc}
        r_1 & 1 && \\
        & r_1 && \\
        && r_2 & \\
        &&& r_3 \\
    \end{array}\right]$ &   $\left[\begin{array}{cccc}
        a_1 & a_2 && \\
        & a_1 && \\
        && b & \\
        &&& c \\
    \end{array}\right]$ &   $\left[\begin{array}{cccc}
        A_1 & A_2 && \\
        A_2 &&& \\
        && B & \\
        &&& C \\
    \end{array}\right]$ \\\hline
    $[ A_{12} ]$  &   $\left[\begin{array}{cccc}
        r & 1 && \\
        & r && \\
        && r & \\
        &&& -3 r \\
    \end{array}\right]$ &   $\left[\begin{array}{cccc}
        a_1 & a_2 & d & \\
        & a_1 && \\
        & e & b & \\
        &&& c \\
    \end{array}\right]$ &   $\left[\begin{array}{cccc}
        A_1 & A_2 & D & \\
        A_2 &&& \\
        D &&B & \\
        &&& C \\
    \end{array}\right]$ \\\hline
    $[ A_{13} ]$  &   $\left[\begin{array}{cccc}
        r &&& \\
        & r && \\
        && r & \\
        &&& -3 r \\
    \end{array}\right]$ &   $\left[\begin{array}{cccc}
        a_{11} & a_{12} & a_{13} & \\
        a_{21} & a_{22} & a_{23} & \\
        a_{31} & a_{32} & a_{33} & \\
        &&& b \\
    \end{array}\right]$ &   $\left[\begin{array}{cccc}
        A_1 & A_2 & A_3 & \\
        A_2 & A_4 & A_5 & \\
        A_3 & A_5 & A_6 & \\
        &&& B \\
    \end{array}\right]$ \\\hline
    $[ A_{14} ]$  &   $\left[\begin{array}{cccc}
        0 & 1 & 0 & \\
        & 0 & 1 & \\
        && 0 & \\
        &&& 0 \\
    \end{array}\right]$ &   $\left[\begin{array}{cccc}
        a_1 & a_2 & a_3 & c \\
        & a_1 & a_2 & 0 \\
        && a_1 & 0 \\
        0 & 0 & d & b
    \end{array}\right]$ &   $\left[\begin{array}{cccc}
        A_1 & A_2 & A_3 & C \\
        A_2 & A_3 && 0 \\
        A_3 &&& 0 \\
        C & 0 & 0 & B
    \end{array}\right]$ \\\hline
    $[ A_{15} ]$  &   $\left[\begin{array}{cccc}
        r & 1 & 0 & \\
        & r & 1 & \\
        && r & \\
        &&& -3 r \\
    \end{array}\right]$ &   $\left[\begin{array}{cccc}
        a_1 & a_2 & a_3 & \\
        & a_1 & a_2 & \\
        && a_1 & \\
        &&& b \\
    \end{array}\right]$ &   $\left[\begin{array}{cccc}
        A_1 & A_2 & A_3 & \\
        A_2 & A_3 && \\
        A_3 &&& \\
        &&& B \\
    \end{array}\right]$ \\\hline
    $[ A_{16} ]$  &   $\left[\begin{array}{cccc}
        r_1 &&& \\
        & r_2 && \\
        && r_3 & \\
        &&& r_4 \\
    \end{array}\right]$ &   $\left[\begin{array}{cccc}
        a &&& \\
        & b && \\
        && c & \\
        &&& d \\
    \end{array}\right]$ &   $\left[\begin{array}{cccc}
        A &&& \\
        & B && \\
        && C & \\
        &&& D \\
    \end{array}\right]$ \\\hline
    $[ A_{17} ]$  &   $\left[\begin{array}{cccc}
        0 & 1 & 0 & 0 \\
        & 0 & 1 & 0 \\
        && 0 & 1 \\
        &&& 0
    \end{array}\right]$ &   $\left[\begin{array}{cccc}
        a_1 & a_2 & a_3 & a_4 \\
        & a_1 & a_2 & a_3 \\
        && a_1 & a_2 \\
        &&& a_1
    \end{array}\right]$ &   $\left[\begin{array}{cccc}
        A_1 & A_2 & A_3 & A_4 \\
        A_2 & A_3 & A_4 & \\
        A_3 & A_4 && \\
        A_4 &&& \\
    \end{array}\right]$ \\\hline
    $[ A_{18} ]$  &   $\left[\begin{array}{cccc}
        r_1 &&& \\
        & r_2 && \\
        && r & -s \\
        && s & r \\
    \end{array}\right]$ &   $\left[\begin{array}{cccc}
        a &&& \\
        & b && \\
        && u & -v \\
        && v & u \\
    \end{array}\right]$ &   $\left[\begin{array}{cccc}
        A &&& \\
        & B && \\
        && U & V \\
        && V & -U \\
    \end{array}\right]$ \\\hline
    $[ A_{19} ]$  &   $\left[\begin{array}{cccc}
        0 & -s & 1 & 0 \\
        s & 0 & 0 & 1 \\
        && 0 & -s \\
        && s & 0
    \end{array}\right]$ &   $\left[\begin{array}{cccc}
        u_1 & -v_1 & u_2 & -v_2 \\
        v_1 & u_1 & v_2 & u_2 \\
        && u_1 & -v_1 \\
        && v_1 & u_1
    \end{array}\right]$ &   $\left[\begin{array}{cccc}
        U_1 & V_1 & U_2 & V_2 \\
        V_1 & -U_1 & V_2 & -U_2 \\
        U_2 & V_2 && \\
        V_2 & -U_2 &&
    \end{array}\right]$ \\\hline
\caption{Equivalence classes of $\mathfrak{sl} ( 4, \mathbb{R} )$.}
\label{table: eq classes sl4}
\end{longtable}

\section{Commutative algebras} %

\begin{longtable}{ c c l}
    $\mathbf{\mathfrak{A}}$ & $\mathbf{e^A}$ & \\\hline\hline
    $\left[\begin{array}{ccc}
        a_1 & a_2 & a_3 \\
        & a_1 & a_2 \\
        && a_1 \\
    \end{array}\right]$ & $\left[\begin{array}{ccc}
        X_1 & X_2 & X_3 \\
        & X_1 & X_2 \\
        && X_1 \\
    \end{array}\right]$ & $\begin{array}{l}
        X_1 = e^{a_1}   \\
        X_2 = e^{a_1} a_2 \\
        X_3 = e^{a_1} ( a_2^2/2 + a_3 )   \\
    \end{array}$ \\\hline
    $\left[\begin{array}{ccc}
        a_1 & a_2 & \\
        & a_1 & \\
        && b \\
    \end{array}\right]$ & $\left[\begin{array}{ccc}
        X_1 & X_2 & \\
        & X_1 & \\
        && Y \\
    \end{array}\right]$ & $\begin{array}{l}
        X_1 = e^{a_1} a_2 \\
        X_2 = e^{a_1} \\
        Y = e^{b} \\
    \end{array}$ \\\hline
    $\left[\begin{array}{ccc}
        a && \\
        & u & -v \\
        & v & u \\
    \end{array}\right]$ & $\left[\begin{array}{ccc}
        X && \\
        & U & -V \\
        & V & U \\
    \end{array}\right]$ & $\begin{array}{l}
        X = A e^{a} \\
        U = e^{u} \cos v\\
        V = e^{u} \sin v \\
    \end{array}$ \\\hline
    $\left[\begin{array}{ccc}
        a_1 && \\
        & a_2 & \\
        && a_3 \\
    \end{array}\right]$ & $\left[\begin{array}{ccc}
        X_1 && \\
        & X_2 & \\
        && X_3 \\
    \end{array}\right]$ & $\begin{array}{l}
        X_1 = e^{a_1} \\
        X_2 = e^{a_2} \\
        X_3 = e^{a_3} \\
    \end{array}$ \\\hline
\caption{Commutative algebras $\mathfrak{A}$ of $\mathfrak{sl} ( 3, \mathbb{R} )$.}
\label{table: comm alg sl3}
\end{longtable}

\begin{longtable}{c c l}
    $\mathbf{\mathfrak{A}}$ & $\mathbf{e^A}$ & \\\hline\hline
    $\left[\begin{array}{cccc}
        a_1 & a_2 && \\
        & a_1 && \\
        && b_1 & b_2 \\
        &&& b_1 \\
    \end{array}\right]$ & $\left[\begin{array}{cccc}
        X_1 & X_2 && \\
        & X_1 && \\
        && Y_1 & Y_2 \\
        &&& Y_1 \\
    \end{array}\right]$ & $\begin{array}{l}
        X_1 = e^{a_1} \\
        X_2 = e^{a_1} a_2 \\
        Y_1 = e^{b_1} \\
        Y_2 = e^{b_1} b_2 \\
    \end{array}$ \\\hline
    $\left[\begin{array}{cccc}
        a_1 & a_2 && \\
        & a_1 && \\
        && b & \\
        &&& c \\
    \end{array}\right]$ & $\left[\begin{array}{cccc}
        X_1 & X_2 && \\
        & X_1 && \\
        && Y_1 & \\
        &&& Y_2 \\
    \end{array}\right]$ & $\begin{array}{l}
        X_1 = e^{a_1} \\
        X_2 = e^{a_1} a_2 \\
        Y_1 = e^{b} \\
        Y_2 = e^{c} \\
    \end{array}$ \\\hline
    $\left[\begin{array}{cccc}
        a_1 & a_2 && \\
        & a_1 && \\
        && u & -v \\
        && v & u \\
    \end{array}\right]$ & $\left[\begin{array}{cccc}
        X_1 & X_2 && \\
        & X_1 && \\
        && U & -V \\
        && V & U \\
    \end{array}\right]$ & $\begin{array}{l}
        X_1 = e ^{a_1} \\
        X_2 =  e ^{a_1} a_2 \\
        U = e^{u} \cos v \\
        V = e^{u} \sin v \\
    \end{array}$ \\\hline
    $\left[\begin{array}{cccc}
        u_1 & -v_1 && \\
        v_1 & u_1 && \\
        && u_2 & -v_2 \\
        && v_2 & u_2 \\
    \end{array}\right]$ & $\left[\begin{array}{cccc}
        U_1 & -V_1 && \\
        V_1 & U_1 && \\
        && U_2 & -V_2 \\
        && V_2 & U_2 \\
    \end{array}\right]$ & $\begin{array}{l}
        U_1 = e ^{u_1} \cos v_1 \\
        V_1 = e ^{u_1} \sin v_1 \\
        U_2 = e ^{u_2} \cos v_2 \\
        V_2 = e ^{u_2} \sin v_2 \\
    \end{array}$ \\\hline
    $\left[\begin{array}{cccc}
        a_1 & a_2 & a_3 & \\
        & a_1 & a_2 & \\
        && a_1 & \\
        &&& b \\
    \end{array}\right]$ & $\left[\begin{array}{cccc}
        X_1 & X_2 & X_3 & \\
        & X_1 & X_2 & \\
        && X_1 & \\
        &&& Y \\
    \end{array}\right]$ & $\begin{array}{l}
        X_1 = e^{a_1}   \\
        X_2 = e^{a_1} a_2 \\
        X_3 = e^{a_1} ( a_2^2/2 + a_3 )   \\
        Y = e^b \\
    \end{array}$ \\\hline
    $\left[\begin{array}{cccc}
        a_1 & a_2 & a_3 & a_4 \\
        & a_1 & a_2 & a_3 \\
        && a_1 & a_2 \\
        &&& a_1 \\
    \end{array}\right]$ & $\left[\begin{array}{cccc}
        X_1 & X_2 & X_3 & X_4 \\
        & X_1 & X_2 & X_3 \\
        && X_1 & X_2 \\
        &&& X_1 \\
    \end{array}\right]$ & $\begin{array}{l}
        X_4 = e^{a_1} ( a_2^3/6 + a_2 a_3 + a_4 ) \\
        X_3 = e^{a_1} ( a_2^2/2 + a_3 ) \\
        X_2 = e^{a_1} a_2 \\
        X_1 = e^{a_1} \\
    \end{array}$ \\\hline
    $\left[\begin{array}{cccc}
        a &&& \\
        & b && \\
        && c & \\
        &&& d \\
    \end{array}\right]$ & $\left[\begin{array}{cccc}
        X_1 &&& \\
        & X_2 && \\
        && X_3 & \\
        &&& X_4 \\
    \end{array}\right]$ & $\begin{array}{l}
        X_1 = e^a \\
        X_2 = e^b \\
        X_3 = e^c \\
        X_4 = e^d  \\
    \end{array}$ \\\hline
    $\left[\begin{array}{cccc}
        a &&&  \\
        & b && \\
        && u & -v \\
        && v & u \\
    \end{array}\right]$ & $\left[\begin{array}{cccc}
        X_1 &&& \\
        & X_2 && \\
        && U & -V \\
        && V & U \\
    \end{array}\right]$ & $\begin{array}{l}
        X_1 = e^a \\
        X_2 = e^b \\
        U = e^u \cos v \\
        V = e^u \sin v \\
    \end{array}$ \\\hline
    $\left[\begin{array}{cccc}
        u_1 & -v_1 & u_2 & -v_2 \\
        v_1 & u_1 & v_2 & u_2 \\
        && u_1 & -v_1 \\
        && v_1 & u_1
    \end{array}\right]$ & $\left[\begin{array}{cccc}
        U_1 & -V_1 & U_2 & -V_2 \\
        V_1 & U_1 & V_2 & U_2 \\
        && U_1 & -V_1 \\
        && V_1 & U_1
    \end{array}\right]$ & $\begin{array}{l}
        U_1 = e^{u_1} \cos v_1 \\
        V_1 = e^{u_1} \sin v_1 \\
        U_2 = e^{u_1} ( u_2 \cos v_1 - v_2 \sin v_1 ) \\
        V_2 = e^{u_1} ( v_2 \cos v_1 + u_2 \sin v_1 ) \\
    \end{array}$ \\\hline
\caption{Commutative algebras $\mathfrak{A}$ of $\mathfrak{sl} ( 4, \mathbb{R} )$.}
\label{table: comm alg sl4}
\end{longtable}

\end{appendices}


\bibliography{sn-bibliography}

\end{document}